\documentclass[11pt,a4paper]{article}
\usepackage[a4paper,hmargin=2.8cm,vmargin=3cm]{geometry}
\usepackage[utf8x]{inputenc}
\usepackage{amsmath,amsthm,amsfonts,amssymb}
\usepackage[bookmarksnumbered=true]{hyperref}
\usepackage{dsfont}
\usepackage{color}
\usepackage{tikz-cd}
\usepackage{authblk}
\usetikzlibrary{positioning}

\let\oldmarginpar\marginpar
\renewcommand\marginpar[1]{\-\oldmarginpar[\raggedleft\footnotesize #1]%
  {\raggedright\footnotesize #1}}

\newtheorem{thm}{Theorem}[section]  % use thm for Theorems to keep numbering consistent
\newtheorem{cor}[thm]{Corollary}
\newtheorem{prp}[thm]{Proposition}
\newtheorem{lem}[thm]{Lemma}

\newtheorem*{rem}{Remark}

\renewcommand{\theequation}{\thesection.\arabic{equation}}

\newcommand{\fock}{\mathcal{F}}		% fock space symbol
\newcommand{\di}{{\textnormal{d}}}		% differential (for integrals)
\newcommand{\Tbb}{\mathbb{T}}

\newcommand{\Zbbn}{\mathbb{Z}^3 \setminus \{0\}}

\DeclareMathOperator\artanh{artanh}

\newcommand{\Hbb}{\mathbb{H}}
\newcommand{\Fcal}{\mathcal{F}}
\newcommand{\Ncal}{\mathcal{N}}		% calligraphic N
\newcommand{\Ocal}{\mathcal{O}}		% big-O, order-of
\newcommand{\hc}{\mbox{h.c.}}		%hermitian conjugate
\newcommand{\pluscc}{\mbox{c.c.}}
\newcommand{\cc}[1]{\overline{#1}}	% complex conjugate
\newcommand{\Rbb}{\mathbb{R}}		% real numbers
\newcommand{\Cbb}{\mathbb{C}}		% complex numbers
\newcommand{\Nbb}{\mathbb{N}}		% natural numbers
\newcommand{\Zbb}{\mathbb{Z}}
\renewcommand{\Re}{\operatorname{Re}\,} 	%RealPart
\newcommand{\id}{\mathbb{I}}
\newcommand{\norm}[1]{\lVert#1\rVert}	%Norm
\newcommand{\tr}{\operatorname{tr}}

\newcommand{\sgn}{\operatorname{sgn}}
\newcommand{\tagg}[1]{ \stepcounter{equation} \tag{\theequation} \label{eq:#1} } % add tag and label in align*-environments

\newcommand{\be}{\begin{equation}}
\newcommand{\ee}{\end{equation}}

\newcommand{\cF}{{\cal F}}

\newcommand{\cN}{{\cal N}}

\newcommand{\new}[1]{#1}

\counterwithin{equation}{section}

\title{Particle--Hole Pair Localization on the Fermi Surface and its Impact on the Correlation Energy}

\author[]{Niels Benedikter}
\affil[]{Università degli Studi di Milano, Via Cesare Saldini 50, 20133 Milano, Italy

\smallskip

External Scientific Member of Basque Center for Applied Mathematics, Alameda de Mazarredo 14, 48009 Bilbao, Bizkaia, Spain}

\affil[]{ORCID: \href{https://orcid.org/0000-0002-1071-6091}{0000-0002-1071-6091}, e--mail: \href{mailto:niels.benedikter@unimi.it}{\textnormal{\nolinkurl{niels.benedikter@unimi.it}}}}

\begin{document}
\maketitle

\begin{abstract}
In recent years it has been shown how approximate bosonization can be used to justify the random phase approximation for the correlation energy of interacting fermions in a mean-field scaling limit. At the core is the interpretation of particle-hole excitations close to the Fermi surface as bosons. The main two approaches however differ in emphasizing collective degrees of freedom (particle-hole pairs delocalized over patches on the Fermi surface) or particle-hole pairs exactly localized in momentum space. Both methods lead to equal precision for the correlation energy with regular interaction potentials. This poses the question how big the influence of delocalizing particle-hole pairs really is. In the present note we show that a description with few, completely collective bosonic degrees of freedom only yields an upper bound of about 92\% of the optimal value. Nevertheless it is remarkable that such a simple approach comes that close to the optimal bound.
\end{abstract}

\section{Introduction and Main Result}
We consider a system of $N$ spinless fermions in three dimensions, with Hamiltonian given in the mean-field scaling with an effective semiclassical parameter as introduced by \cite{NS81}:
\[
  H_N := -\hbar^2 \sum_{i=1}^N \Delta_{x_i} + \frac{1}{N}\sum_{i<j}^N V(x_i - x_j), \qquad \hbar := N^{-1/3}.
\]
We assume that the system is restricted to the cube $[0,2\pi]^3$ with periodic boundary conditions, or more precisely, $H_N$ acts as a self-adjoint operator on the anti-symmetrized Hilbert space of square-integrable functions on the torus $L^2_a\left((\Tbb^3)^N\right)$.
 The ground state energy is
\[
  E_N := \inf_{\substack{\psi \in L^2_a\left((\Tbb^3)^N\right)\\\norm{\psi} = 1}} \langle \psi, H_N \psi \rangle.
\]
By restricting to Slater determinants ($S_N$ denotes the symmetric group)
\[
  \psi_\text{Slater} (x_1, \dots , x_N) = \frac{1}{\sqrt{N!}} \sum_{\pi \in S_N} \sgn(\pi) f_1 (x_{\pi(1)}) f_2 (x_{\pi(2)}) \cdots f_N (x_{\pi(N)})
\]
with  $\{f_j\}_{j=1}^N$ being an orthonormal set in $L^2 (\Tbb^3)$, we obtain the Hartree-Fock functional
\begin{equation}\label{eq:hfenergy}
\begin{split}
  \langle \psi_\text{Slater}, H_N \psi_\text{Slater} \rangle = \mathcal{E}_\text{HF}(\omega) & := \tr \, \left(- \hbar^2 \Delta \right) \omega + \frac{1}{N} \int_{\Tbb^3 \times \Tbb^3} \di x \di y V(x-y) \omega (x;x) \omega (y;y) \\
  & \hspace{2.74cm}- \frac{1}{N} \int_{\Tbb^3 \times \Tbb^3} \di x \di y V(x-y) |\omega (x;y)|^2 \;,
\end{split}
\end{equation}
with the projection (the one-particle reduced density matrix of the Slater determinant)
\[
  \omega = \sum_{j=1}^N \lvert f_j \rangle \langle f_j\rvert \;.
\]
The first term containing $V$ in \eqref{eq:hfenergy} is called the direct term, the second the exchange term.
Of course, minimizing the Hartree-Fock functional over all orthonormal sets generally only produces an upper bound on the ground state energy.

While Hartree-Fock theory is much easier than the full many-body problem, it is nevertheless in general not possible to explicitly find the minimizers. It is exceptional that in the present setting the set of $N$ orthonormal plane waves
\[
  f_k(x) := (2\pi)^{-3/2} e^{ikx}, \quad k \in \Zbb^3
\]
with smallest wave vectors $\lvert k\rvert$ (i.\,e., minimizing the kinetic energy)
not only constitutes a stationary point of the functional, but the global minimizer \cite[Appendix A]{BNP+21}. (In fact, in the thermodynamic limit, this is not the case: the ground state tends to spontaneusly break symmetries \cite{GL19}, including translation invariance by the formation of density waves. The energy difference to the plane wave state however remains exponentially small \cite{GHL19}.) In momentum space, this can be visualized as the Fermi ball, a ball around the origin in $\mathbb{Z}^3$ with radius chosen such that it contains $N$ points of the lattice $\Zbb^3$. The radius is called the Fermi momentum $k_\textnormal{F}$ and is of order $N^{1/3}$.

The energy difference between the many-body ground state energy and the minimum of the Hartree-Fock functional is called the correlation energy. It has first been computed by formal methods known as \emph{random phase approximation} by \cite{Gellmann,sawada,sawadabruckner}. Recently rigorous bosonization methods have been developed by \cite{BNP+20,Ben21,BNP+21,BPSS23} and \cite{CHN22,CHN23,CHN23a,CHN24} to provide a rigorous justification of the random phase approximation. Both approaches rely on treating particle-hole pair excitations as approximately bosonic particles; they differ however in the first approach considering particle-hole pairs delocalized in a superposition of all kinematically admissible states near a patch on the Fermi surface, the second approach instead considering individual particle-hole pairs with sharply defined momenta. However, given that both approaches produce the same results for the leading order of the correlation energy, it is natural to ask how big the effect of localizing (in patches or sharply) particle-hole pairs on the Fermi surface really is. Maybe one could even avoid any localization and a description in terms of particle-hole pairs completely delocalized over the Fermi surface might be sufficient?

In the present note, we optimize over random phase approximation trial states, however restricting to completely delocalized particle-hole pairs, that is, of the form
\begin{equation}\label{eq:rpa_state}
\psi_\Xi := R_\omega T\Omega \;, \qquad T := \exp\bigg(\frac{1}{2}\sum_{k\in \Zbbn} \Xi(k) b^*_k b^*_{-k} - \hc \bigg)\;,
\end{equation}
where $R_\omega$ the particle-hole transformation corresponding to the Hartree--Fock minimizer $\omega$, and $T$  is a quasi-bosonic Bogoliubov transformation given in terms of the collective particle-hole pair creation operators on fermionic Fock space
\[
  b^*_k  := \frac{1}{n_k} \sum_{\substack{p \in B_\textnormal{F}^c\\h \in B_\textnormal{F}}} \delta_{p-h,k} a^*_p a^*_h \;.
\]
(The operators $a^*_p$ are fermionic creation operators for a particle excitation with momentum $p$ outside the non-interacting Fermi ball $B_\textnormal{F}$, the operators $a^*_p$ are fermionic creation operators for a hole excitation with momentum $h$ inside the Fermi ball; $n_k$ is a normalization constant, and $\Omega$ the vacuum vector in Fock space. A detailed discussion is given in Section~\ref{sec:two}.)

Our result shows that the best possible such trial state cannot reproduce even the leading order of the correlation energy correctly; surprisingly though, one can get quite close (to about 92\% of the correlation energy that has been proven in the earlier mentioned papers to be the correct leading order).
\begin{thm}[Main Result]\label{thm:mainresult}
Assume that $\inf_{k \in \Zbb^3}\hat{V}(k) \geq - \frac{1}{4}\big(\frac{16}{9\pi}\big)^{2/3}$ and that the quantities $A_1$ through $A_3$ given in Lemma \ref{lem:collectederrors} are all finite. (For example, $\hat{V}$ with compact support.) Then we have
\begin{align*}
  \inf_{\Xi} \langle \psi_\Xi, H_N \psi_\Xi \rangle = \mathcal{E}_\textnormal{HF}(\omega) - \frac{\hbar}{2} \sum_{k \in \Zbbn} \lvert k\rvert \left(\mu + \mu^{-1} \hat{V}(k) - \sqrt{\mu^2 - 2 \hat{V}(k)} \right) + \mathcal{O}(N^{-1/2}) \;,
\end{align*}
where $\omega$ is the projection onto the $N$ plane waves $f_k$ with smallest $\lvert k\rvert$ ($k \in \Zbb^3$), and $\mu := \left( \frac{16}{9 \pi} \right)^{1/3}$.
The infimum is over all choices of $\Xi(k)$ in states of the form \eqref{eq:rpa_state}.
\end{thm}
The Hartree--Fock energy $\mathcal{E}_\text{HF}(\omega)$ is of order $N$; or more precisely, its kinetic and direct parts are of order $N$ and its exchange term is of order $1$.
For comparison, in the following theorem we recall the correct value of the correlation energy. Up to second order in the interaction potential, this was first proven by \cite{HPR20}. The formula including all orders of the potential was proven by \cite{BNP+20,BNP+21,BPSS23} based on approximate bosonization of collective pair excitations. An non-collective approximate bosonization was then developed by \cite{CHN23a}, which ultimately enabled the inclusion of a further exchange-type term in the energy for singular potentials \cite{CHN23,CHN24}. (Recently a computation of the correct order of the correlation energy has been obtained also for particles interacting more strongly than in the mean-field regime \cite{FRS25}.)

In reading the proof of the main result, it is interesting to note that the error terms due to Fock space estimates (the many-body analysis) are all of order $N^{-1}$; only the error terms due to the explicit computation of $n_k^2$, $\alpha_k$, and $\beta_k$ (defined as sums over lattice points in $\Zbb^3$, compare \eqref{eq:alphabeta}) contribute an error at order $N^{-2/3}$. Counting lattice points in convex bodies is a much studied subject, and improved lattice point counting estimates would directly improve the error term in our main theorem, and we conjecture it to be as small as $N^{-1 + \varepsilon}$.
\begin{thm}[Optimal Energy]
Let $\kappa := (\frac{3}{4\pi})^{1/3}$. There exists $\alpha >0$ such that
\[
\begin{split}
  E_N & = \mathcal{E}_\textnormal{HF}(\omega)\\
  & \quad + \hbar \kappa \sum_{k \in \Zbb^3} \lvert k\rvert \left[ \frac{1}{\pi} \int_0^\infty  \log\left( 1+ 2 \pi \kappa \hat{V}(k) \left(1-\lambda \arctan\frac{1}{\lambda} \right) \right) \di\lambda - \frac{\pi}{2}\kappa \hat{V}(k) \right]\\
  & \quad + \Ocal(N^{-1/3-\alpha})\;.
\end{split}
\]
If the interaction potential is weak we have
\[
  E_N = \mathcal{E}_\textnormal{HF}(\omega) - \hbar \frac{\pi}{2} \left(1-\log(2)\right) \sum_{k \in \Zbb^3} \lvert k \rvert \hat{V}(k)^2 + \hbar \mathcal{O}\Big(\big(\hat{V}(k)\big)^3\Big) + \Ocal(N^{-1/3-\alpha})\;.
\]
% Notice that $1-\log(2) \simeq 0.3069$.
\end{thm}
The general formulas including all orders of the interaction potential are difficult to compare analytically, but fixing a choice of $\hat{V}(k)$, they can easily be compared numerically. Alternatively, we consider the case of weak interaction potential so that we can expand to second order in $\hat{V}(k)$.
\begin{cor}[Second Order Expansion]
To second order in $\hat{V}(k)$, we have
\[
  \inf_{\Xi} \langle \psi_\Xi, H_N \psi_\Xi \rangle = \mathcal{E}_\textnormal{HF}(\omega) - \hbar  \frac{\pi}{2} \frac{9}{32} \sum_{k\in \Zbbn} \lvert k\rvert \hat{V}(k)^2  + \hbar \mathcal{O}\Big(\big(\hat{V}(k)\big)^3\Big)+ \mathcal{O}(N^{-1/2})\;.
\]
\end{cor}
At second order, this is about 92\% of the optimal energy, where instead of $\frac{9}{32} = 0.28125$, the pre-factor is $(1-\log(2)) \simeq 0.3068$.

\section{Particle-Hole Transformation and Hartree-Fock Theory} \label{sec:two}
We consider $L^2_a\left((\Tbb^3)^N\right)$ as embedded in the fermionic Fock space constructed over $L^2\left(\Tbb^3\right)$. The creation and annihilation operators (ore more precisely, operator-valued distributions) for $x \in \Tbb^3$ are defined in the standard way and denoted by $a^*_x$ and $a_x$, respectively. Since we consider a fermionic systems, they satisfy the canonical anticommutation rules (CAR)
\[
   \{ a_x, a_y \} = 0 = \{a^*_x, a^*_y\} \;, \qquad \{ a_x, a^*_y \} = \delta(x-y) \;.
\]
Given any rank-$N$ projection $\omega = \sum_{j=1}^N \lvert f_j \rangle \langle f_j \rvert$ (which we think of as a one-particle reduced density matrix of a Slater determinant) we define
\[
   u := 1 - \omega \;, \qquad v  := \sum_{j=1}^N \lvert \cc{f_j} \rangle \langle f_j \rvert \;.
\]
Following the theory of fermionic Bogoliubov transformations \cite{Sol14}, one can define the particle-hole transformation $R_\omega$ as a unitary on Fock space such that
\begin{equation}\label{eq:phtrafo}
   R_\omega\Omega = \psi_\textnormal{Slater}\;, \qquad R_\omega a_x R_\omega^* = a(u_x) + a^*(v_x) \;, \qquad R_\omega a^*_x R_\omega^* = a^*(u_x) + a(v_x)\;.
\end{equation}
Here we introduced the notation $u_x := u(\cdot,x)$ and $v_x := v(\cdot,x)$, where $u(\cdot,\cdot)$, and $v(\cdot, \cdot)$ denote the (distributional) integral kernels of the operators $u$ and $v$. The detailed construction of the particle-hole transformation can be found, for example, in \cite{BD23,BPS}. We proceed to compute the transformed Hamiltonian, where for all the rest of the paper we fix $\omega$ as the projection onto the $N$ plane waves with lowest kinetic energy. Using \eqref{eq:phtrafo} and the CAR to re-arrange the result in normal order, that is, with all creation operators to the left of all the annihilation operators, we obtain
\[
  R^*_\omega H_N R_\omega = \mathcal{E}_\text{HF}(\omega) + \di\Gamma(uhu-\cc{v}\cc{h}v) + \int_{\Tbb^3 \times \Tbb^3} \di x \di y \left(a^*_x a^*_y (uh\cc{v})(x,y) + \hc\right) + Q_N \;.
\]
In this formula $h$ is the Hartree-Fock Hamiltonian (which depends on $\omega$, see \eqref{eq:hfhamiltonian} for the explicit form). If $\omega$ projects onto a stationary point of the Hartree-Fock functional, then $uh\cc{v} = 0$, so the $(a^* a^* + a a)$-term vanishes.

Consequently we are looking for a trial state $\xi \in \fock$ (containing equal numbers of particles and holes, equivalent to $R_\omega\xi$ being an $N$-particle state) such that
\begin{equation}\label{eq:phtransform}
  \langle R_\omega \xi, H_N R_\omega \xi \rangle = \mathcal{E}_\text{HF}(\omega) + \langle \xi, \left( \di\Gamma(uhu-\cc{v}\cc{h}v) + Q_N\right)\xi\rangle < \mathcal{E}_\text{HF}(\omega) \;.
\end{equation}
The quartic terms are
\[
\begin{split}
  Q_N & = \frac{1}{2N} \int_{\Tbb^3 \times \Tbb^3} \di x\di y V(x-y) \bigg(\mathcal{E}_1  + 2 a^*(u_x) a^*(\cc{v}_x) a(\cc{v}_y) a(u_y)\\
  & \hspace{5cm} + \Big[ a^*(u_x) a^*(u_y) a^*(\cc{v}_y) a^*(\cc{v}_x) + \mathcal{E}_2  + \hc\Big]
  \bigg)
\end{split}
\]
where
\begin{align*}
  \mathcal{E}_1 & = a^*(u_x)a^*(u_y) a(u_y)a(u_x)- 2 a^*(u_x) a^*(\cc{v}_y) a(\cc{v}_y) a(u_x) + a^*(\cc{v}_y)a^*(\cc{v}_x) a(\cc{v}_x) a(\cc{v}_y) \;, \\
  \mathcal{E}_2 & = - 2a^*(u_x) a^*(u_y)a^*(\cc{v}_x) a(u_y) + 2 a^*(u_x) a^*(\cc{v}_y) a^*(\cc{v}_x) a(\cc{v}_y) \;.
\end{align*}
The term $\mathcal{E}_1$ can be estimated using the number operator, while $\mathcal{E}_2$ does not contribute to the expectation value due to a parity argument (it creates particles in $\pm1$-steps, but the trial state that we use contains only multiples of $2$ particles or holes).
\begin{lem}\label{lem:errornoparthole}
 For all $\xi \in \fock$ we have
 \[
  \lvert \langle \xi, \frac{1}{2N}\int_{\Tbb^3\times \Tbb^3} \di x\di y\, V(x-y) \,\mathcal{E}_1\, \xi\rangle\rvert \leq \frac{2}{N} \sum_{k \in \Zbb^3} \lvert \hat{V}(k)\rvert\, \langle \xi, \Ncal^2\xi\rangle \;.
 \]
\end{lem}
\begin{rem}
We use the convention that, for example, $\langle \xi, \Ncal^2 \xi\rangle$ is read as $+\infty$ if $\xi$ is not in the domain of $\Ncal^2$. The inequality is then true but trivial. For our taste this leads to a more readable text than specifying the operator domains in all inequalities.
 \end{rem}
\begin{proof}[Proof of Lemma~\ref{lem:errornoparthole}] We use the Fourier decomposition $V(x-y) = \sum_{k \in \Zbb^3} \hat{V}(k) e^{ik(x-y)}$ and Lemma \ref{lem:classics} to get a bound in terms of the operator norm of any bounded operator $A$, namely
 \[
  \lvert \langle \xi, \di\Gamma(A) \xi\rangle\rvert \leq \norm{A} \langle \xi, \Ncal \xi\rangle \;. \qedhere
 \]
\end{proof}

\begin{lem}\label{lem:errorparity}
Let $T$ be an operator that commutes with $i^\Ncal$, then
\[
  \langle T\Omega,\mathcal{E}_2 T\Omega\rangle =0 \;.
\]
\end{lem}
\begin{proof}
 We can insert an $i^\Ncal$ in the right argument, since $i^\Ncal \Omega = i^0 \Omega = \Omega$:
 \begin{align*}
  \langle T\Omega, \mathcal{E}_2 T\Omega\rangle & = \langle T \Omega, \mathcal{E}_2 T i^\Ncal \Omega \rangle = \langle T\Omega, \mathcal{E}_2 i^\Ncal T\Omega \rangle \\
  & = \langle T\Omega, i^{\Ncal-2} \mathcal{E}_2 T\Omega \rangle  = - \langle T (-i)^\Ncal \Omega, \mathcal{E}_2 T\Omega\rangle = - \langle T\Omega, \mathcal{E}_2 T\Omega \rangle\;. \qedhere
 \end{align*}
\end{proof}
We continue with
\begin{equation}
\label{eq:splitinteraction}
  Q_N = Q_N^{(0)} + \frac{1}{2N} \int_{\Tbb^3 \times \Tbb^3} \di x\di y V(x-y) \big(\mathcal{E}_1 + \mathcal{E}_2 + \mathcal{E}_2^* \big)
\end{equation}
where
\[
\begin{split}
  Q_N^{(0)} & := \frac{1}{2N} \sum_{k \in \Zbb^3} \hat{V}(k) \int_{\Tbb^3 \times \Tbb^3} \di x \di y \bigg( 2 a^*(u_x)e^{ikx} a^*(\cc{v}_x) a(\cc{v}_y) e^{-iky} a(u_y)\\
& \hspace{5cm}+ \left[a^*(u_x)e^{ikx} a^*(\cc{v}_x) a^*(u_y) e^{-iky} a^*(\cc{v}_y)  + \hc\right]\bigg) \;.\end{split}
\]
We introduce global bosonic particle-hole pair excitations by
\begin{equation}
\label{eq:bosonicpair}
  \tilde b^*_k := \int_{\Tbb^3} \di x\, a^*(u_x) e^{ikx} a^*(\cc{v}_x) = \sum_{\substack{p \in B_\textnormal{F}^c\\h \in B_\textnormal{F}}} \delta_{p-h,k} a^*_p a^*_h \;.
\end{equation}
Here $B_\textnormal{F}$ is the Fermi ball and $B_\textnormal{F}^c = \Zbb^3 \setminus B_\textnormal{F}$. Due to $uv=0$, we have $\tilde{b}^*_0 = 0$. Then
\[
  Q_N^{(0)} = \frac{1}{2N} \sum_{k\in \Zbbn} \hat{V}(k) \left( 2 \tilde b^*_k \tilde b_k + \tilde b^*_k \tilde b^*_{-k} + \tilde b_{-k} \tilde b_k \right) \;.
\]
We normalize the new operators to strengthen the similarity to bosonic creation and annihilation operators, introducing
\begin{equation}\label{eq:almostbosonic}
  b^*_k := \frac{1}{n_k} \tilde{b}^*_k \;, \qquad n_k^2 := \norm{\tilde b^*_k\Omega}^2 \;.
\end{equation}
\begin{lem}[Normalization Constant]\label{lem:normalizationconstant}
Let $\hat{x}$ be the position operator, that is, multiplication by the coordinate $x$.
 We have
 \[
  n_k = \sqrt{\tr e^{ik\hat{x}}(1-\omega)e^{-ik\hat{x}} \omega} = \norm{u e^{ik\hat{x}} \cc{v}}_\textnormal{HS}\;.
 \]
 This can be computed explicitly, yielding
  \begin{equation}
 \label{eq:normalconst}
  n_k^2 = n_{-k}^2 = \left( \frac{9\pi}{16}\right)^{1/3} \lvert k\rvert N \hbar +  \mathcal{O}(N^{1/2}) \;, \quad n_k = \left( \frac{3}{4} \sqrt{\pi}\right)^{1/3} \sqrt{\lvert k\rvert N \hbar} + \mathcal{O}(N^{1/6})\;.
 \end{equation}
\end{lem}
\begin{proof}
Using the CAR it is straightforward to find
\[
  \begin{split}
    n_k^2 & = \norm{\tilde b^*_k \Omega}^2 = \langle \Omega, \int \di x a(\cc{v}_x) e^{-ikx} a(u_x) \int \di y a^*(u_y) e^{iky} a^*(\cc{v}_y) \Omega\rangle\\
   & = \int \di x \di y e^{-ikx} \langle u_x,u_y\rangle e^{iky} \langle \cc{v}_x,\cc{v}_y \rangle  = \tr e^{-ik\hat{x}} (1-\omega) e^{ik\hat{x}} \omega
  \end{split}
\]
where we used $\langle u_x,u_y\rangle = (1-\omega)(x,y)$ and $\langle \cc{v}_x,\cc{v}_y\rangle = \omega(y,x)$.

To show that $n_k$ is even as a function of $k$, we write
\begin{align*}
  n_k^2 & = \tr e^{-ik\hat{x}}(1-\omega)e^{ik\hat{x}} \omega = \tr \omega - \tr e^{-ik\hat{x}} \omega e^{ik\hat{x}} \omega\\
 n_{-k}^2 & = \tr e^{ik\hat{x}} (1-\omega)e^{-ik\hat{x}} \omega = \tr \omega - \tr e^{ik\hat{x}} \omega e^{-ik\hat{x}}\omega \;.
\end{align*}
By cyclicity of the trace, the last expressions of the lines are the same.

For the computation of $n_k^2$, consider Figure~\ref{fig:normalization}: the trace of a projection is just its rank, which corresponds to the number of points of $\Zbb^3$ in the dark area in the figure.
The volume of the overlap of two balls, both with radius $R$, with centers displaced by a distance $d$ is $V_\text{lense} = \frac{\pi}{12} (4R+d)(2R-d)^2$.
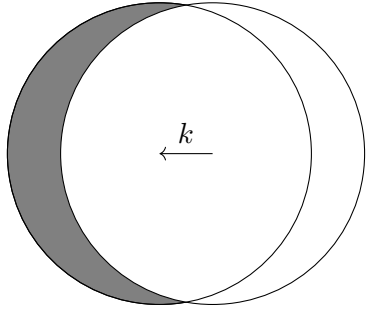
\begin{figure}\centering
\begin{tikzpicture}
 \draw[fill=gray] (4,2) circle (2cm);
 \draw[fill=white] (4.7,2) circle (2cm);
 \draw (4,2) circle (2cm);
 \node [above] at (4.35,2) {$k$};
 \draw [->] (4.7,2)--(4,2);
\end{tikzpicture}
 \caption{The two balls represent the projections $e^{ikx}\omega e^{-ikx}$ and $\omega$, respectively. The normalization constant $n_k^2$ is given by the number of lattice points in the lense marked in gray.}\label{fig:normalization}
\end{figure}
The Fermi ball contains $N$ modes; thus $\frac{4}{3}\pi R^3 = N$, yielding radius $k_\textnormal{F} = \left(\frac{3}{4\pi} N\right)^{1/3} + \mathcal{O}(N^0)$. The volume of the shaded region in Figure \ref{fig:normalization} is
\[
\begin{split}
  V_\text{ball} - V_\text{lense} & = \frac{4}{3}\pi k_\textnormal{F}^3 - \frac{\pi}{12}(4 k_\textnormal{F}+\lvert k\rvert)(2 k_\textnormal{F}-\lvert k\rvert)^2 = \pi k_\textnormal{F}^2 \lvert k\rvert - \frac{\pi}{12}\lvert k\rvert^3 \\
  & = \lvert k\rvert \left( \frac{3}{4} \sqrt{\pi}\right)^{2/3} N \hbar - \frac{\pi}{12} \lvert k\rvert^3 \;.
\end{split}
\]
By Gauss' classical argument for the circle problem, the leading order of each of the lense and the ball volumes capture the order $k_\text{F}^3$ of the respective number of lattice points. To justify the difference taken, we have to go to the next order and use better remainder estimates, since Gauss' remainder estimates are of order $k_F^2$ just like the term we are interested in.

It was shown \cite{Hea99} that the number of lattice points in a ball of radius $R$ is
\begin{equation}
   \lvert \{ x \in \Zbb^3 : \lvert x\rvert \leq R \} \rvert = \frac{4\pi}{3} R^3 + \mathcal{O}_\varepsilon(R^{\frac{29}{22} + \varepsilon}) \;.
\end{equation}
For the number of points in the lense we have to refer to results on the number of lattice points in large convex bodies \cite{Hla50,Mul99,Guo12}. If we refer to \cite{Hla50} as the simplest result going beyond Gauss' simple counting, the difference between the volume of the lense and the number of lattic points it contains is of order
\[
   n_k^2 = V_\text{ball} - V_\text{lense} + \mathcal{O}(k_\textnormal{F}^{3/2}) =  \lvert k\rvert \left( \frac{3}{4} \sqrt{\pi}\right)^{2/3} N \hbar + \mathcal{O}(N^{1/2}) \;.
\]
To be precise, the mentioned result applies to convex bodies with $C^\infty$-boundary. The corners of the lense therefore have to be smoothened. The number of lattice points counted incorrectly due to this smoothing can be captured by a Gauss-type argument as order $k_F$, and therefore subleading compared to the error of order $k_\text{F}^{3/2}$ that we already have. The error term for $n_k$ follows from the one for $n_k^2$ by Taylor expansion.
\end{proof}

We conclude that
\begin{equation}
\label{eq:quadrhamilt}
  Q_N^{(0)} = \frac{1}{2N} \sum_{k\in \Zbbn} \hat{V}(k) n_k^2 \left( 2 b^*_k b_k + b^*_k b^*_{-k} + b_{-k} b_k \right)\;.
\end{equation}
We do not have a simple formula for $\di\Gamma(uhu-\cc{v}hv)$ in terms of $b^*_k$ and $b_k$, but we will see below that we can calculate its expectation value in the trial state anyway.

\section{Almost-Bosonic Collective Operators}
The operators $b_k$ and $b^*_k$ in good approximation behave like bosonic creation and annihilation operators satisfying the CCR.

\begin{lem}[Almost-CCR]\label{lem:approxccr}
Let $k, l \in \Zbb^3$. The $b$-operators annihilate the vacuum $\Omega \in \fock$,
\begin{equation}\label{eq:annihilation}b_k \Omega = 0.\end{equation}
Furthermore they satisfy the approximate canonical commutator relations
 \begin{equation}\label{eq:approxccr}
  [b_{k},b^*_{l}]  = \delta_{k,l} + \mathcal{E}(k,l) \;, \qquad
  [b^*_{k},b^*_{l}]  = [b_{k},b_{l}] = 0\;,
 \end{equation}
 where the error term $\mathcal{E}(k,l)$ is an operator that can be estimated using the fermionic particle number operator $\Ncal$ by
\begin{equation}
\label{eq:ccrerror}
  \norm{\mathcal{E}(k,l) \xi} \leq \frac{1}{n_k n_l} \norm{\Ncal \xi} \quad \text{for all } \xi \in \fock\;.
\end{equation}
Note also that $\mathcal{E}(k,l)^* = \mathcal{E}(l,k)$.
\end{lem}
\begin{proof}
Straight-forward computations using the fermionic CAR, c.\,f.,~\cite{BNP+20}.
\end{proof}

We need the following estimates; a proof can be found, e.\,g., in \cite{BPS}.)
\begin{lem}\label{lem:classics}
For every bounded operator $O$, we have
\[ \| d\Gamma (O) \psi \| \leq \| O \| \, \| \cN \psi \| \]
for every $\psi \in \cF$. If $O$ is a Hilbert-Schmidt operator, we also have the bounds
\begin{equation} \label{eq:HS-bds}
\begin{split}
  \left\| \int dx dx' \, O(x;x') a_x a_{x'} \psi \right\| &\leq \| O \|_{\text{HS}} \, \| \cN^{1/2} \psi \|, \\
  \left\| \int dx dx' \, O(x;x') a^*_x a^*_{x'} \psi \right\| & \leq 2 \| O \|_{\text{HS}} \, \| (\cN+1)^{1/2} \psi \| \;.
\end{split}
\end{equation}
\end{lem}
These bounds directly imply estimates for the almost-bosonic operators $b_k$ and $b^*_k$.
\begin{lem}
For every $k \in \Zbb^3$ we have, for all $\psi \in \fock$, the estimates
\begin{equation}\label{eq:bstark}
  \norm{b_k \psi} \leq \norm{\Ncal^{1/2}\psi} \;,
 \qquad
\norm{b^*_k \psi} \leq 2\norm{(\Ncal+1)^{1/2}\psi} \;.
\end{equation}
\end{lem}
\begin{proof}
These bounds follow directly by using Lemma \ref{lem:classics} with the definition of the pair operators \eqref{eq:bosonicpair} and recalling the Hilbert-Schmidt and trace norms from Lemma \ref{lem:normalizationconstant}.
\end{proof}

\section{Almost-Bosonic Quasifree Trial State}
The quadratic expression for the interaction given in \eqref{eq:quadrhamilt} suggests that the we should think of the correlation energy as originating from a quadratic almost-bosonic Hamiltonian. Thus we take our trial state as a state of the form $\exp(b^* b^* - bb) \Omega$, imitating the form of a bosonic quasifree state but using the almost-bosonic pair operators instead of exactly bosonic operators.  We are going to calculate the expectation value of the Hamiltonian in a state of this from and then optimize over the function $\Xi$ which appears in the exponent and in the exactly bosonic case serves to parametrize the Bogoliubov transformations.

More precisely, we define the following unitary operator
\begin{equation}
\label{eq:bogtransform}
  T(\lambda) := \exp(\lambda B), \quad B := \frac{1}{2}\sum_{k\in \Zbbn} \Xi(k) b^*_k b^*_{-k} - \hc
\end{equation}
which mimics a bosonic Bogoliubov transformation, but replacing bosonic creation operators by our quasi-bosonic pair operators. For simplicity we write $T := T(1)$. Note that $T(0) = \id$. Our trial state $T\Omega$ has equal numbers of particles and holes,
\begin{equation}\label{eq:phnumber}
   \bigg(\sum_{p \in B_\textnormal{F}^c} a^*_p a_p - \sum_{h \in B_\textnormal{F}} a^*_h a_h \bigg) T \Omega = 0
\end{equation}
because \eqref{eq:bogtransform} contains only pair operators, which, by \eqref{eq:bosonicpair}, always create or annihilate a particle and a hole together (so the operator in parenthesis in \eqref{eq:phnumber} commutes with $T$). Applying the particle-hole transformation explicitly to \eqref{eq:phnumber}, this shows that $R_\omega T \Omega$ is an eigenvector of the total particle number operator $\Ncal$ with eigenvalue $N$.

\begin{lem}[Almost-Bosonic Bogoliubov Transformation]\label{lem:almostbogtrafo}
Let $\Xi(k) = \Xi(-k)$. For $l \neq 0$, the conjugation of $b_l$ and $b^*_l$ with $T$ is given by
 \begin{equation}\label{eq:bogtrafo}
 \begin{split}
  T^* b_l T & = \cosh(\Xi)(l) b_l + \sinh(\Xi)(l) b^*_{-l} + \mathcal{E}_k(\Xi),\\
  T^* b^*_l T & = \cc{\cosh(\Xi)(l)} b^*_l + \cc{\sinh(\Xi)(l)} b_{-l} + \mathcal{E}^*_k(\Xi) \;,
 \end{split}\end{equation}
 where we set
\[
  \cosh(\Xi)(l) := 1 - \frac{1}{2!} \Xi(l) \cc{\Xi}(-l) + \cdots, \qquad \sinh(\Xi)(l) := \Xi(l) - \frac{1}{3!} \Xi(l)\cc{\Xi(-l)} \Xi(l) + \cdots
\]
The operators $b_0$ and $b^*_0$ are invariant.
The operators $\mathcal{E}_k(\Xi)$ can be estimated by
\begin{equation}\label{eq:bogtrafoerror}
  \norm{\mathcal{E}_l(\Xi)\psi} \leq 4 \sup_{t \in [0,1]} \norm{(\Ncal+2)^{3/2}T(t)\psi} e^{\lvert \Xi(l)\rvert} \sum_{k \in \Zbbn}  \frac{\lvert \Xi(k)\rvert}{n_k n_l}\;.
\end{equation}
The same bound holds for the adjoint $\mathcal{E}_l(\Xi)^*$.
\end{lem}
Since $\Xi(k) = \Xi(-k)$, we have
\begin{equation}\label{eq:absvalue}
  \cosh(\Xi)(l) = \cosh(\lvert \Xi(l)\rvert) \quad \text{and} \quad \sinh(\Xi)(l) = \sinh(\lvert \Xi(l)\rvert) \frac{\Xi(l)}{\lvert \Xi(l)\rvert}\;.
\end{equation}
\begin{proof}
We have
\[
\begin{split}
  T^* b_l T - b_l & = \int_0^1 \di \lambda\, \frac{\di}{\di\lambda} \left( e^{-\lambda B} b_l e^{\lambda B} \right) = \int_0^1 \di\lambda\, T(\lambda)^* [b_l,B] T(\lambda)\\
  & = \int_0^1 \di \lambda\, T(\lambda)^* \frac{1}{2} \sum_{k \in \Zbbn} \Xi(k) \left( [b_l,b^*_k] b^*_{-k} + b^*_k [b_l,b^*_{-k}] \right) T(\lambda) \\
  & = \int_0^1 \di \lambda\, T(\lambda)^* \left( \Xi(l) b^*_{-l} + \frac{1}{2}\sum_{k \in \Zbbn} \Xi(k) \left( \mathcal{E}(k,l)b^*_{-k} + b^*_{k} \mathcal{E}(-k,l) \right)\right) T(\lambda)
\end{split}
\]
implying
\begin{equation}\label{eq:btransform}
\begin{split}
  T^* b_l T & = b_l + \Xi(l) \int_0^1 \di \lambda\, T(\lambda)^* b^*_{-l} T(\lambda) \\
  & \quad + \frac{1}{2}\sum_{k \in \Zbbn} \Xi(k) \int_0^1 \di\lambda\, T(\lambda)^* \left( \mathcal{E}(k,l) b^*_{-k} + b^*_k \mathcal{E}(-k,l) \right) T(\lambda)
\end{split}
\end{equation}
and in the same way
\[
\begin{split}
  T^* b^*_l T - b^*_l & = \int_0^1 \di \lambda\, T(\lambda)^* \frac{1}{2} \sum_{k \in \Zbbn} \cc{\Xi(k)} \left( [b^*_l,b_{-k}] b_{k} + b_{-k} [b^*_l,b_k] \right) T(\lambda) \\
  & = \int_0^1 \di \lambda\, T(\lambda)^* \left( -\cc{\Xi(l)} b_{-l} - \frac{1}{2}\sum_{k \in \Zbbn} \cc{\Xi(k)} \left( \mathcal{E}(-k,l)b_{k} + b_{-k} \mathcal{E}(k,l) \right)\right) T(\lambda)
\end{split}
\]
implying
\begin{equation}\label{eq:bstartransform}
\begin{split}
  T^* b^*_l T & = b^*_l - \cc{\Xi(l)} \int_0^1 \di \lambda \, T(\lambda)^* b_{-l} T(\lambda) \\
  & \quad - \frac{1}{2}\sum_{k \in \Zbbn} \cc{\Xi(k)} \int_0^1 \di\lambda\, T(\lambda)^* \left( \mathcal{E}(-k,l) b_k + b_{-k} \mathcal{E}(k,l) \right) T(\lambda) \;.
\end{split}
\end{equation}
Plugging \eqref{eq:btransform} and \eqref{eq:bstartransform} into each other iteratively (we iterate only in the leading term, the error terms containing $\mathcal{E}$ we do not iterate), we arrive at
\[
\begin{split}
  T^* b_l T & = \cosh(\Xi)(l) b_l + \sinh(\Xi)(l) b^*_{-l} + \mathcal{E}_\text{sinh}(\Xi)(l) + \mathcal{E}_\text{cosh}(\Xi)(l)\\
  & =: \cosh(\Xi)(l) b_l + \sinh(\Xi)(l) b^*_{-l} + \mathcal{E}_l(\Xi)\;,
\end{split}
\]
where the error terms are
\[
\begin{split}
  \mathcal{E}_\text{cosh}(\Xi)(l) & = \sum_{k \in \Zbb^3\setminus\{0\}} \frac{\cc{\Xi(k)}}{2} \Bigg(
  - \Xi(l) \int_0^1 \di \lambda\int_0^\lambda \di\lambda' T(\lambda')^* A_\text{cosh}(k,l) T(\lambda') \\ & \quad + \Xi(l)\cc{\Xi(-l)}\Xi(l) \int_0^1 \!\di\lambda \int_0^\lambda \!\di \lambda' \int_0^{\lambda'}\! \di \lambda'' \int_0^{\lambda''}\! \di \lambda''' T(\lambda''')^* A_\text{cosh}(k,l) T(\lambda''')
  - \cdots
  \Bigg)
\end{split}
\]
and
\[
\begin{split}
  \mathcal{E}_\text{sinh}(\Xi)(l) & = \sum_{k \in \Zbbn} \frac{\Xi(k)}{2} \Bigg(
  \int_0^1 \di \lambda T(\lambda)^* A_\text{sinh}(k,l) T(\lambda) \\
  & \quad - \Xi(l) \cc{\Xi(-l)} \int_0^1 \di \lambda \int_0^\lambda \di\lambda' \int_0^{\lambda'} \di \lambda'' T(\lambda'')^* A_\text{sinh}(k,l) T(\lambda'')
  + \cdots  \Bigg)
\end{split}
\]
with
\begin{equation}\label{eq:Asinh}
  A_\text{cosh}(k,l) = \mathcal{E}(-k,-l) b_k + b_{-k} \mathcal{E}(k,-l), \qquad A_\text{sinh}(k,l) = \mathcal{E}(k,l) b^*_{-k} + b^*_{k} \mathcal{E}(-k,l) \;.
\end{equation}
(The tail term of the expansion vanishes as the expansion order tends to infinity.)

We proceed to estimate $\mathcal{E}_l(\Xi)$. First of all, for all $\psi \in \fock$ we have
\[
  \norm{\mathcal{E}_l(\Xi)\psi} \leq \norm{\mathcal{E}_\text{cosh}(\Xi)(l) \psi} + \norm{\mathcal{E}_\text{sinh}(\Xi)(l) \psi} \;.
\]
Using the triangle inequality
\[
\begin{split}
  \norm{\mathcal{E}_\text{cosh}(\Xi)(l)\psi} & \leq \frac{1}{2}\sum_{k \in \Zbb^3\setminus\{0\}} \lvert{\Xi(k)}\rvert \Bigg(
  \lvert \Xi(l)\rvert \int_0^1 \di \lambda\int_0^\lambda \di\lambda' \norm{A_\text{cosh}(k,l) T(\lambda') \psi} \\ & \qquad + \lvert \Xi(l)\rvert^3 \int_0^1\di\lambda \int_0^\lambda \di \lambda' \int_0^{\lambda'} \di \lambda'' \int_0^{\lambda''} \di \lambda''' \norm{A_\text{cosh}(k,l) T(\lambda''')\psi}
  + \cdots \Bigg)
\end{split}
\]
We now give an estimate independent of $\lambda$ for the norms on the right hand side of the previous equation, using \eqref{eq:bstark} and the fact that $\Ncal$ commutes with $\mathcal{E}(k,-l)$ since $\mathcal{E}(k,-l)$ is a sum of two operators $\di\Gamma(uAu)$ and $\di\Gamma(\cc{v}Av)$ ($\di\Gamma$-operators conserve the number of particles):
\begin{align*}
  \norm{A_\text{cosh}(k,l) T(\lambda)\psi}&= \norm{   \left(\mathcal{E}(-k,-l) b_k + b_{-k} \mathcal{E}(k,-l)\right)T(\lambda)\psi}\\
  & \leq \frac{1}{n_k n_l} \norm{\Ncal b_k T(\lambda)\psi} + \norm{\Ncal^{1/2} \mathcal{E}(k,-l) T(\lambda)\psi}\\
  & \leq \frac{1}{n_k n_l} \norm{(\Ncal+2)b_kT(\lambda)\psi} + \norm{ \mathcal{E}(k,-l)\Ncal^{1/2} T(\lambda)\psi}\\
  & = \frac{1}{n_k n_l} \norm{b_k\Ncal T(\lambda)\psi} + \frac{1}{n_k n_l}\norm{\Ncal \Ncal^{1/2} T(\lambda)\psi}\\
  & \leq \frac{1}{n_k n_l} \norm{\Ncal^{1/2}\Ncal T(\lambda)\psi} + \frac{1}{n_k n_l} \norm{\Ncal^{1/2}\Ncal T(\lambda)\psi}\\
  & \leq \frac{2}{n_k n_l} \langle T(\lambda) \psi, \Ncal^3 T(\lambda) \psi \rangle^{1/2}\;.
\end{align*}
Thus
\[
\begin{split}
  \norm{\mathcal{E}_\text{cosh}(\Xi)(l)\psi} & \leq \sum_{k \in \Zbb^3\setminus\{0\}} \lvert{\Xi(k)}\rvert \sup_{t\in [0,1]} \frac{1}{n_k n_l} \langle T(t) \psi, \Ncal^3 T(t) \psi \rangle^{1/2}\\
  & \quad \times \Bigg(
  \lvert \Xi(l)\rvert \int_0^1 \di \lambda\int_0^\lambda \di\lambda' + \lvert \Xi(l)\rvert^3 \int_0^1\di\lambda \int_0^\lambda \di \lambda' \int_0^{\lambda'} \di \lambda'' \int_0^{\lambda''} \di \lambda'''
  + \cdots
  \Bigg) \\
  & \leq \sum_{k \in \Zbb^3\setminus\{0\}} \lvert{\Xi(k)}\rvert\frac{1}{n_k n_l} \sup_{t\in [0,1]}  \langle T(t) \psi, \Ncal^3 T(t) \psi \rangle^{1/2}\Bigg(
  \lvert \Xi(l)\rvert \frac{1}{2!} + \lvert \Xi(l)\rvert^3 \frac{1}{4!}+ \cdots
  \Bigg)\;.
\end{split}
\]
The series in the big parenthesis is summable.

In the same way as the estimates above we find
\[
  \norm{A_\text{sinh}(k,l) T(\lambda)\psi} \leq \frac{4}{n_k n_l} \langle T(\lambda) \psi, (\Ncal+2)^3 T(\lambda) \psi \rangle
\]
and from that
\[
  \norm{\mathcal{E}_\text{sinh}(\Xi)(l)\psi} \leq 2 \sum_{k \in \Zbbn} \lvert \Xi(k)\rvert \frac{1}{n_k n_l} \sup_{t\in [0,1]} \langle T(t) \psi, (\Ncal+2)^3 T(t) \psi\rangle^{1/2} \left(1 + \frac{1}{3!} \lvert \Xi(l) \rvert^2 + \ldots \right)
\]
Combining, we find \eqref{eq:bogtrafoerror}.

To obtain the same bound for the adjoint, observe that the error terms will have the same form, the essential difference being that instead of \eqref{eq:Asinh} its adjoint appears. This however does not change the estimates we used above: $\mathcal{E}(-k,-l)^*$ can be bounded in the same way as $\mathcal{E}(-k,-l)$ (recall \eqref{eq:ccrerror}), and $b^*_k$ can be bounded in the same way as $b_k$ by \eqref{eq:bstark}, except for a factor $2$ and $(\Ncal + 1)^{1/2}$ appearing in the place of $\Ncal^{1/2}$.
\end{proof}

\section{Calculating the Interaction Energy}
The expectation value of the interaction energy $Q_N^{(0)}$ is easy to calculate now.
\begin{prp}[Interaction Energy]\label{lem:interaction}
We have
\[
\begin{split}
  \langle T \Omega, Q_N^{(0)} T \Omega \rangle & = \frac{1}{N} \Re \!\!\!\sum_{k \in \Zbbn}\!\!\! \hat{V}(k) n_k^2 \left( \sinh(\lvert\Xi(k)\rvert)^2 + \frac{\cc{\Xi(k)}}{\lvert \Xi(k)\rvert} \sinh(\lvert \Xi(k)\rvert) \cosh(\lvert \Xi(k)\rvert)\right) + \varepsilon_1
\end{split}
\]
where the error $\varepsilon_1 \in \Cbb$ can be estimated by
\[
\begin{split}
  \lvert \varepsilon_1 \rvert & \leq \frac{2}{N} \sum_{k \in \Zbbn} \hat{V}(k) n_k^2 \Bigg[ 4\sup_{t \in [0,1]} \langle T(t) \Omega, (\Ncal+2)^3 T(t) \Omega \rangle e^{2\lvert \Xi(k)\rvert} \left(\sum_{m \in \Zbb^3 \setminus\{0\}} \frac{\lvert \Xi(m)\rvert}{n_m n_k}\right)^2 \\
  &\quad +  \big(4\sinh(\lvert \Xi(k)\rvert)+2 \cosh(\lvert \Xi(k)\rvert) \big)  \sup_{t \in [0,1]} \norm{(\Ncal+2)^{3/2}T(t)\Omega} e^{\lvert \Xi(k)\rvert} \sum_{m \in \Zbbn}  \frac{\lvert \Xi(m)\rvert}{n_m n_k} \Bigg] \;.
\end{split}
\]
\end{prp}
\begin{proof}
We apply the transformation rule \eqref{eq:bogtrafo} and then the annihilation of the vacuum \eqref{eq:annihilation} to obtain
\[
\begin{split}
  & \langle T \Omega, Q_N^{(0)} T\Omega \rangle \\
  & = \frac{1}{2N} \!\sum_{k \in \Zbb^3\setminus\{0\}}\!\! \hat{V}(k) n_k^2 \langle \Omega,\Bigg[ 2 \left( \cc{\cosh(\Xi)(k)}b^*_k + \cc{\sinh(\Xi)(k)} b_{-k} + \mathcal{E}_k^*(\Xi) \right)\\
  & \hspace{5.5cm} \times \Big( \cosh(\Xi)(k) b_k + \sinh(\Xi)(k) b^*_{-k} + \mathcal{E}_k(\Xi) \Big)\\
  & \hspace{4.2cm} + \Bigg(\left( \cc{\cosh(\Xi)(k)}b^*_k + \cc{\sinh(\Xi)(k)} b_{-k} + \mathcal{E}^*_k(\Xi)\right)\\
  & \hspace{5.2cm} \times \left(\cc{\cosh(\Xi)(-k)} b^*_{-k} + \cc{\sinh(\Xi)(-k)} b_k + \mathcal{E}_{-k}^*(\Xi)\right) + \hc \Bigg) \Bigg] \Omega\rangle \\
  & = \frac{1}{2N} \sum_{k \in \Zbb^3\setminus\{0\}} \hat{V}(k) n_k^2 \langle \Omega, \left( \cc{\sinh(\Xi)(k)} b_{-k} + \mathcal{E}_k^*(\Xi) \right) \Big( \sinh(\Xi)(k) b^*_{-k} + \mathcal{E}_k(\Xi) \Big)\\
  & \hspace{11em} + \left( \cc{\sinh(\Xi)(k)} b_{-k} + \mathcal{E}^*_k(\Xi)\right)\left(\cc{\cosh(\Xi)(-k)} b^*_{-k} + \mathcal{E}_{-k}^*(\Xi)\right) \Omega\rangle + \pluscc
\end{split}
\]
Recalling that $\norm{b^*_{-k}\Omega}^2 = 1$, we extract the term stated in the Lemma (use \eqref{eq:absvalue} to simplify). The remain terms are errors, given by
\begin{align}
  & \frac{1}{2N} \sum_{k \in \Zbbn} \hat{V}(k) n_k^2 \Bigg( \norm{\mathcal{E}_k(\Xi)\Omega}^2 + \left( \cc{\sinh(\Xi)(k)} \langle b^*_{-k} \Omega, \mathcal{E}_k(\Xi) \Omega\rangle + \pluscc \right) + \langle \mathcal{E}_k^*(\Xi) \Omega, \mathcal{E}^*_{-k}(\Xi)\Omega\rangle \nonumber \\
  & \hspace{3.5cm} + \cc{\sinh(\Xi)(k)} \langle b^*_{-k}\Omega, \mathcal{E}^*_{-k}(\Xi)\Omega\rangle + \cc{\cosh(\Xi)(-k)} \langle \mathcal{E}_k(\Xi) \Omega, b^*_{-k} \Omega \rangle \Bigg) + \pluscc \nonumber \\
  & =: \varepsilon_1 \;.  \label{eq:errorepsone}
\end{align}
\paragraph{Estimating the Error Terms.}
We now show that $\varepsilon_1$ is smaller than the leading term (of order $\hbar = N^{-1/3}$) of the correlation energy. From Lemma~\ref{lem:almostbogtrafo} we have
\begin{equation}
\begin{split}
  & \left\lvert \frac{1}{2N} \sum_{k \in \Zbbn} \hat{V}(k) n_k^2 \norm{\mathcal{E}_k(\Xi) \Omega}^2 \right\rvert\\
  & \leq \frac{2}{N} \sum_{k \in \Zbb^3\setminus\{0\}} \hat{V}(k) n_k^2 \sup_{t \in [0,1]} \langle T(t) \Omega, (\Ncal+2)^3 T(t) \Omega \rangle e^{2\lvert \Xi(k)\rvert} \left(\sum_{m \in \Zbb^3 \setminus\{0\}} \frac{\lvert \Xi(m)\rvert}{n_m n_k}\right)^2 \;.
\end{split}
\end{equation}
In the same way
\begin{equation}
\begin{split}
  & \left\lvert \frac{1}{2N} \sum_{k \in \Zbbn} \hat{V}(k) n_k^2 \langle \mathcal{E}_k^*(\Xi) \Omega, \mathcal{E}^*_{-k}(\Xi)\Omega\rangle \right\rvert \\
  & \leq \frac{4}{N} \sum_{k \in \Zbb^3\setminus\{0\}} \hat{V}(k) n_k^2 \sup_{t \in [0,1]} \langle T(t) \Omega, (\Ncal+2)^3 T(t) \Omega \rangle e^{2\lvert \Xi(k)\rvert} \left(\sum_{m \in \Zbb^3 \setminus\{0\}} \frac{\lvert \Xi(m)\rvert}{n_m n_k}\right)^2 \;.
\end{split}
\end{equation}
Furthermore, using \eqref{eq:bogtrafoerror} and $\norm{b^*_{-k}\Omega}^2 = 1$, we have
\begin{equation}
\begin{split}
& \left\lvert \frac{1}{2N} \sum_{k \in \Zbbn} \hat{V}(k) n_k^2 \left( \cc{\sinh(\Xi)(k)} \langle b^*_{-k} \Omega, \mathcal{E}_k(\Xi) \Omega\rangle + \pluscc \right) \right\rvert \\
& \leq \frac{1}{N} \sum_{k \in \Zbbn} \hat{V}(k) n_k^2 \lvert \sinh(\Xi)(k)\rvert \norm{b^*_{-k}\Omega} \norm{\mathcal{E}_k(\Xi)\Omega}\\
& \leq \frac{2}{N} \sum_{k \in \Zbbn} \hat{V}(k) n_k^2 \sinh(\lvert \Xi(k)\rvert)  \sup_{t \in [0,1]} \norm{(\Ncal+2)^{3/2}T(t)\Omega} e^{\lvert \Xi(k)\rvert} \sum_{m \in \Zbbn}  \frac{\lvert \Xi(m)\rvert}{n_m n_k} \;.
\end{split}
\end{equation}
The terms on the last line of the definition of $\varepsilon_1$, \eqref{eq:errorepsone}, can be controlled in the same way,
\begin{equation}
\begin{split}
  & \left\lvert \frac{1}{2N} \sum_{k \in \Zbbn} \hat{V}(k) n_k^2 \cc{\sinh(\Xi)(k)} \langle b^*_{-k} \Omega, \mathcal{E}^*_{-k}(\Xi) \Omega\rangle  \right\rvert \\
  & \leq \frac{2}{N} \sum_{k \in \Zbbn} \hat{V}(k) n_k^2 \sinh(\lvert \Xi(k)\rvert)  \sup_{t \in [0,1]} \norm{(\Ncal+2)^{3/2}T(t)\Omega} e^{\lvert \Xi(k)\rvert} \sum_{m \in \Zbbn}  \frac{\lvert \Xi(m)\rvert}{n_m n_k}
\end{split}
\end{equation}
and
\begin{equation}
\begin{split}
& \left\lvert \frac{1}{2N} \sum_{k \in \Zbbn} \hat{V}(k) n_k^2 \cc{\cosh(\Xi)(-k)} \langle \mathcal{E}_k(\Xi) \Omega, b^*_{-k} \Omega \rangle \right\rvert\\
& \leq \frac{1}{N} \sum_{k \in \Zbbn} \hat{V}(k) n_k^2 {\cosh(\lvert \Xi(k)\rvert)} \sup_{t \in [0,1]} \norm{(\Ncal+2)^{3/2}T(t)\Omega} e^{\lvert \Xi(k)\rvert} \sum_{m \in \Zbbn}  \frac{\lvert \Xi(k)\rvert}{n_m n_k} \;.
\end{split}
\end{equation}
The $+\pluscc$ add an overall factor $2$.
\end{proof}

\section{Calculating the Kinetic Energy}
Calculating the expectation value of the kinetic energy requires a little more work since we do not have an approximation as a quadratic form of almost-bosonic operators. We start by estimating the corrections to the kinetic energy from the Hartree-Fock operator, which turn out to be small. The Hartree-Fock operator is
\begin{equation}\label{eq:hfhamiltonian}
h \simeq -\hbar^2 \Delta + D + X \;, \quad D = \int_{\Tbb^3} V(x)\di x \;, \quad \norm{X}  \leq \frac{C}{N} \;,
\end{equation}
where we use that for plane waves the direct term $D$ is a constant cancelling out from the differences, and the exchange term $X$ is small in operator norm.
\begin{lem}
Let $X(x,y) := N^{-1}V(x-y) \omega(x,y)$ be the integral kernel of the exchange operator, with
\[
  \omega(x,y) = \sum_{h \in B_\textnormal{F}} e^{ih\cdot(x-y)} =: g(x-y) \;.
\]
Then its operator norm is bounded by
\[
  \norm{X} \leq \frac{1}{N} (2\pi)^{-3/2} \norm{\hat{V}}_{L^1(\Tbb^3)}\;.
\]
\end{lem}
\begin{proof}
The integral kernel is translation invariant, $X(x,y) = N^{-1} V(x-y) g(x-y)$, so $X$ is a convolution operator, which in Fourier space is a multiplication operator. By unitarity
\[
  \norm{X} = \norm{\Fcal X \Fcal^{-1}} = \norm{\hat{X}}_{L^\infty(\Zbb^3)} \;.
\]
Due to the convolution theorem we have
\[
  \hat{X}(k) = \frac{1}{N}\widehat{(Vg)}(k) = \frac{1}{N}(2\pi)^{-d/2} \hat{V}\ast\hat{g}(k)
\]
and thus
\[
  \norm{\hat{X}}_{L^\infty(\Zbb^3)} =  \frac{1}{N}(2\pi)^{-d/2} \sup_{k \in \Zbb^3}\big\lvert \sum_{l \in \Zbb^3} \hat{g}(k-l) \hat{V}(l) \big\rvert \leq \frac{1}{N}(2\pi)^{-d/2} \sup_{k \in \Zbb^3} \lvert \hat{g}(k) \rvert \sum_{l \in \Zbb^3}  \lvert \hat{V}(l)\rvert \;.
\]
Furthermore, $g$ is given as the Fourier transform of the Fermi ball,
\[
  g = \check{\chi}, \quad \chi (k) = \sum_{h \in B_\textnormal{F}} \delta_{h,k} \;,
\]
so inverting the Fourier transform and using that the Kronecker deltas are bounded by one (recall that we are on a torus) $ \sup_{k \in \Zbb^3} \lvert \hat{g}(k) \rvert =1$.
\end{proof}
We now conclude that direct and exchange term do not contribute.
\begin{lem}[Direct and Exchange Term]\label{lem:freekinetic}
Let
\begin{equation}  \label{eq:Hbb0}
  \Hbb_0 := \di\Gamma\left(u(-\hbar^2\Delta)u - \cc{v}(-\hbar^2\Delta)v\right) \;.
\end{equation}
 Let $\xi$ be the almost-bosonic quasifree trial state. Then we have
 \[
  \langle \xi, \di\Gamma(uhu-\cc{v}\cc{h}v) \xi \rangle = \langle \xi, \Hbb_0 \xi \rangle + \mathcal{E}_\text{X} \;,
\]
 where the error term satisfies
 \[
  \lvert \mathcal{E}_\text{X} \rvert \leq \frac{2(2\pi)^{-3/2} \norm{\hat{V}}_{L^1(\Tbb^3)}}{N} \langle \xi,\Ncal \xi \rangle \;.
  \]
\end{lem}
\begin{proof}
 Since the direct term $D$ is just a number, we have
 \[
  \di\Gamma(uDu-\cc{v}Dv) = D\bigg(\sum_{p \in B_\textnormal{F}^c} a^*_p a_p - \sum_{h\in B_\textnormal{F}} a^*_h a_h \bigg) \;.
 \]
 The expectation value of this operator vanishes as discussed in \eqref{eq:phnumber}.
 
 To estimate the exchange term we use the inequality $\lvert \langle \xi, \di\Gamma(A) \xi \rangle\rvert \leq \norm{A} \langle \xi,\Ncal \xi\rangle$ valid for any bounded operator $A$ as well as the previous lemma, so that we obtain
 \[
 \begin{split}
 \lvert \langle \xi, \di\Gamma(uXu-\cc{v}Xv) \xi\rangle \rvert & \leq 2\left(\norm{uXu}+\norm{\cc{v}Xv}\right) \langle \xi,\Ncal \xi\rangle\leq 2 \norm{X} \langle \xi,\Ncal \xi\rangle \\
 & \leq \frac{2(2\pi)^{-3/2} \norm{\hat{V}}_{L^1(\Tbb^3)}}{N} \langle \xi,\Ncal \xi\rangle \;,
 \end{split}
 \]
 where we used that the operator norms are $\norm{u} = 1 = \norm{v}$.
\end{proof}

To compute the kinetic energy, we rely on a Duhamel expansion; this requires us to first compute the commutator of the kinetic energy operator with a pair creation operator.
\begin{lem}[Kinetic Energy Commutator]\label{lem:kinencomm}
Let $\Hbb_0$ be defined as in \eqref{eq:Hbb0}. We have
\[
  \left[ \Hbb_0, b^*_k\right] = \frac{1}{n_k} \sum_{\substack{p\in B_\textnormal{F}^c\\h\in B_\textnormal{F}}} \delta_{p,h+k}\hbar^2 \left(p^2-h^2\right)a^*_p a^*_h = \hbar^2 k\cdot c^*_k \;,
\]
with the vector-valued operator $c^*_k$ given by
\[
  c^*_k := \frac{1}{n_k} \sum_{\substack{p\in B_\textnormal{F}^c\\h\in B_\textnormal{F}}} \delta_{p,h+k}  \left(p+h\right)a^*_p a^*_h \;.
\]
Likewise
\[
  \left[\Hbb_0,b_k\right] = -\hbar^2 k \cdot c_k \;.
\]
\end{lem}
\begin{proof}
It is convenient to first verify that in momentum representation
\[
  \Hbb_0 = \sum_{p \in B^c_\textnormal{F}} \hbar^2 p^2 a^*_p a_p - \sum_{h \in B_\textnormal{F}} \hbar^2 h^2 a^*_h a_h \;.
\]
The rest of the proof is then a straight-forward computation using the CAR.
\end{proof}

The following lemma permits us to estimate error terms in the computation of the kinetic energy.

\begin{lem}\label{lem:kingkong}
With $c^*_k$ as defined in Lemma \ref{lem:kinencomm}, we have
\[
  [c^*_k,b_l] = - \delta_{k,l} f(k) + \mathcal{E}_\text{c}(k,l)
\]
with the function $f: \Zbb^3 \to \Rbb^3$ given by
\begin{equation}\label{eq:kinnumeric}
  f(k) := \frac{1}{n_k^2} \sum_{\substack{p \in B_\textnormal{F}^c\\ h\in B_\textnormal{F}}} \delta_{p,h+k}(p+h) \;.
\end{equation}
For any vector $m \in \Rbb^3$, the error operator $\mathcal{E}_\text{c}(k,l)$ is bounded by
\[
  \norm{m\cdot\mathcal{E}_\text{c}(k,l) \psi} \leq  \new{ \lvert m\rvert  \frac{4 \Big(\left(\frac{3}{4\pi} \right)^{1/3} N^{1/3} + \lvert k\rvert \Big)}{n_k n_l}} \norm{\Ncal \psi} \quad \text{for all } \psi \in \fock \;.
\]
The same holds for $\mathcal{E}^*_\text{c}(k,l)$. Moreover $\mathcal{E}_\text{c}(k,l)$ commutes with the number operator $\Ncal$.
\end{lem}
\begin{proof} Using the CAR we find
\begin{align*}
  [c^*_k,b_l] & = \frac{1}{n_k n_l} \sum_{\substack{p \in B_\textnormal{F}^c\\h \in B_\textnormal{F}}} \delta_{p,h+k} (p+h) \sum_{\substack{\tilde p \in B_\textnormal{F}^c\\ \tilde h \in B_\textnormal{F}}} \delta_{\tilde p,\tilde h+l} [a^*_p a^*_h, a_{\tilde h} a_{\tilde p}] \\
  & = \frac{1}{n_k n_l} \sum_{\substack{p, \tilde p \in B_\textnormal{F}^c\\ h,\tilde h \in B_\textnormal{F}}} \delta_{p,h+k} (p+h) \delta_{\tilde p,\tilde h + l} \left( -\delta_{h,\tilde h} \delta_{p,\tilde p} + \delta_{h, \tilde h} a^*_p a_{\tilde p} + \delta_{p, \tilde p} a^*_h a_{\tilde h} \right) \\
  & = - \frac{1}{n_k n_l} \delta_{k,l} \sum_{\substack{p\in B_\textnormal{F}^c\\ h \in B_\textnormal{F}}} \delta_{p,h+k} (p+h) \\
  & \quad + \frac{1}{n_k n_l} \sum_{\substack{p,\tilde p \in B_\textnormal{F}^c\\h \in B_\textnormal{F}}} \delta_{p,h+k} (p+h) \delta_{\tilde p, h+l} a^*_p a_{\tilde p} + \frac{1}{n_k n_l} \sum_{\substack{p\in B_\textnormal{F}^c\\h,\tilde h \in B_\textnormal{F}}} \delta_{p,h+k} (p+h) \delta_{p, \tilde h+l} a^*_h a_{\tilde h} \\
  & =:  - \delta_{k,l} f(k) + \mathcal{E}_\text{c}(k,l)  \; ,
\end{align*}
in the last line defining the error operator $\mathcal{E}_\text{c}(k,l)$. The first term of $m\cdot \mathcal{E}_\text{c}(k,l)$ can be written
\[
  \frac{1}{n_k n_l} \sum_{\substack{p,\tilde p \in B_\textnormal{F}^c\\h \in B_\textnormal{F}}} \delta_{p,h+k} m\cdot (p+h) \delta_{\tilde p, h+l} a^*_p a_{\tilde p} = \frac{1}{n_k n_l} \di\Gamma(A(k,l,m)) \;,
\]
where
\[
\begin{split}
  A(k,l,m)_{p,\tilde p} & := \sum_{h \in B_\textnormal{F}} \delta_{p,h+k} \delta_{\tilde p, h+l} m\cdot (p+h) \chi(p \in B_\textnormal{F}^c) \chi(\tilde{p} \in B_\textnormal{F}^c) \\
  & = m\cdot(2p-k) \chi(p-k \in B_\textnormal{F}) \delta_{\tilde p, p-k+l} \chi(p \in B_\textnormal{F}^c) \chi(\tilde{p} \in B_\textnormal{F}^c) \;.
\end{split}
\]
We then have the usual estimate using the operator norm of $A(k,l,m)$,
\[
  \norm{\di\Gamma(A(k,l,m))\psi} \leq \norm{A(k,l,m)} \norm{\Ncal \psi} \;.
\]
The operator norm of $A$ can be controlled as follows:
\begin{equation*}
\begin{split}
  \norm{A(k,l,m)} & = \sup_{\substack{x\in \ell^2(\Zbb^3)\\\norm{x}=1}} \norm{A(k,l,m) x} = \sup_{\substack{x\in \ell^2(\Zbb^3)\\\norm{x}=1}} \left(\sum_{p\in B_\textnormal{F}^c} \left( \sum_{\tilde p \in B_\textnormal{F}^c} A(k,l,m)_{p,\tilde p} x_{\tilde p}\right)^2 \right)^{1/2} \\
  & = \sup_{\substack{x\in \ell^2(\Zbb^3)\\\norm{x}=1}} \left( \sum_{p \in B_\textnormal{F}^c} \bigg(m\cdot(2p-k)\chi(p-k\in B_\textnormal{F})\chi(p-k+l \in B_\textnormal{F}^c) x_{p-k+l}\bigg)^2\right)^{1/2} \\
  & = \sup_{\substack{x\in \ell^2(\Zbb^3)\\\norm{x}=1}} \left( \sum_{p \in B_\textnormal{F}^c} \big(m\cdot(2p-k)\big)^2\chi(p-k\in B_\textnormal{F})\chi(p-k+l \in B_\textnormal{F}^c) x^2_{p-k+l}\right)^{1/2} \\
  & \leq \sup_{\substack{x\in \ell^2(\Zbb^3)\\\norm{x}=1}} \left( \sum_{p \in B_\textnormal{F}^c} \lvert m\rvert^2 (2\lvert p\rvert + \lvert k\rvert)^2\chi(p-k\in B_\textnormal{F})\chi(p-k+l \in B_\textnormal{F}^c) x^2_{p-k+l}\right)^{1/2} \;.
\end{split}
\end{equation*}
Due to the constraint $p - k \in B_\textnormal{F}$, we have $\lvert p\rvert \leq \left( \frac{3}{4\pi} \right)^{1/3} N^{1/3} + \lvert k \rvert$. We obtain
\begin{align}
  \norm{A(k,l,m)} & \leq \lvert m\rvert \bigg(2 \Big(\frac{3}{4\pi} \Big)^{1/3} N^{1/3} + 3 \lvert k\rvert \bigg) \sup_{\substack{x\in \ell^2(\Zbb^3)\\\norm{x}=1}} \left( \sum_{p \in B_\textnormal{F}^c} x^2_{p-k+l}\right)^{1/2} \label{eq:Abound1} \\
  & \leq \lvert m\rvert \bigg(2 \Big(\frac{3}{4\pi} \Big)^{1/3} N^{1/3} + 3 \lvert k\rvert \bigg)  \sup_{\substack{x\in \ell^2(\Zbb^3)\\\norm{x}=1}} \norm{x} = \lvert m\rvert \bigg(2 \Big(\frac{3}{4\pi} \Big)^{1/3} N^{1/3} + 3 \lvert k\rvert \bigg) \;. \nonumber
\end{align}

The second term of $m\cdot \mathcal{E}_\text{c}(k,l)$ can be estimated by an analogous computation: the differences are insignificant; for example, instead of the factor $(2\lvert p\rvert + \lvert k\rvert)^2$ we have $(2 \lvert h\rvert + \lvert k\rvert)^2$, which by means of $h \in B_\textnormal{F}$ in this case yields
\begin{align*}
  \norm{A(k,l,m)} \leq \lvert m\rvert \bigg(2 \Big(\frac{3}{4\pi} \Big)^{1/3} N^{1/3} + \lvert k\rvert \bigg) \;.
\end{align*}
The analogous bound for $\mathcal{E}^*_\text{c}(k,l)$ amounts to an exchange of $p$ with $\tilde{p}$, or $h$ with $\tilde{h}$, respectively.
Obviously $ \mathcal{E}_\text{c}(k,l)$ commutes with $\Ncal$ because it can be written as a $\di\Gamma$--operator.
\end{proof}

We can now proceed to computing the expectation value of the kinetic energy.

\begin{prp}[Expectation of Kinetic Energy]\label{lem:kinetic}
We have
\[
  \langle T\Omega, \Hbb_0 T\Omega \rangle= \sum_{k \in \Zbbn}  \hbar^2 k\cdot f(k) {\sinh(\lvert \Xi(k)\rvert)^2}  + 2 \Re \varepsilon_2 \;,
\]
where
\[
  f(k) := \frac{1}{n_k^2} \sum_{\substack{p \in B_\textnormal{F}^c\\ h\in B_\textnormal{F}}} \delta_{p,h+k}(p+h)
\]
and there exists a constant $C$
such that the error term is controlled by
\[
\begin{split}
  \lvert \varepsilon_2 \rvert & \leq 8 \hbar^2\!\!\!\! \sum_{k \in \Zbbn} \! \Bigg[ 4\lvert \Xi(k)\rvert \lvert k\cdot f(k)\rvert \sinh(\lvert \Xi(k)\rvert) \sup_{t \in [0,1]} \norm{(\Ncal+2)^{3/2}T(t)\Omega} e^{\lvert \Xi(k)\rvert} \!\!\!\sum_{m \in \Zbbn}\!\!\!  \frac{\lvert \Xi(m)\rvert}{n_m n_k} \\
  &  \hspace{2.5cm} + \sum_{l \in \Zbbn} \lvert \Xi(l)\rvert \new{ \frac{8 \Big(\left(\frac{3}{4\pi} \right)^{1/3} N^{1/3} + \lvert k\rvert \Big)}{n_k n_l} \lvert k \rvert } \sup_{\mu \in [0,1]} \norm{ (\Ncal+1)^{3/2}  T(\mu)\Omega} \\
  & \hspace{3.5cm} \times \Big(  \sinh(\lvert \Xi(k)\rvert) + \sup_{t \in [0,1]} \norm{(\Ncal+2)^{3/2}T(t)\Omega} e^{\lvert \Xi(k)\rvert} \!\!\!\sum_{m \in \Zbbn} \!\!  \frac{ \lvert \Xi(m)\rvert}{n_m n_k}\Big) \\
  & \hspace{2.5cm} \new{ + 2 \lvert k\cdot f(k)\rvert \sup_{t \in [0,1]} \langle T(t) \Omega, (\Ncal+2)^{3}T(t)\Omega\rangle e^{2\lvert \Xi(k)\rvert} \Bigg( \sum_{m \in \Zbbn}  \frac{\lvert \Xi(m)\rvert}{n_m n_k} \Bigg)^2 } \Bigg]\,.
\end{split}
\]
\end{prp}
\begin{proof}
Noting that $\Hbb_0 \Omega = 0$, by Duhamel we obtain
\begin{align*}
  & \langle T\Omega, \Hbb_0 T\Omega \rangle = \int_0^1 \di \lambda \langle T(\lambda)\Omega, \left[ \Hbb_0, B\right] T(\lambda)\Omega \rangle \\
  & = \int_0^1 \di \lambda \frac{1}{2} \sum_{k \in \Zbbn} \Xi(k) \langle T(\lambda)\Omega, \left( [\Hbb_0,b^*_k] b^*_{-k} + b^*_k [\Hbb_0,b^*_{-k}] \right) T(\lambda)\Omega \rangle \\
  & \quad - \int_0^1 \di\lambda \frac{1}{2} \sum_{k \in \Zbbn} \cc{\Xi(k)} \langle T(\lambda)\Omega, \left( [\Hbb_0,b_{-k}]b_k + b_{-k} [\Hbb_0,b_k]\right) T(\lambda)\Omega \rangle \\
  & = \int_0^1 \di \lambda \frac{1}{2} \sum_{k\in \Zbbn}\Xi(k)\hbar^2 \langle T(\lambda)\Omega, \left( k\cdot c^*_k b^*_{-k} + b^*_k (-k)\cdot c^*_{-k} \right)T(\lambda)\Omega\rangle + \pluscc \\
  & = \int_0^1 \di \lambda \sum_{k\in \Zbbn}\Xi(k)\hbar^2 k\cdot  \langle T(\lambda)\Omega, c^*_k b^*_{-k} T(\lambda)\Omega\rangle + \pluscc \tagg{dontforgetcc}
\end{align*}
where in the last step we used that $\Xi(k) = \Xi(-k)$ and that $c^*_k$ and $b^*_{-k}$ commute since they consist of pairs of creation operators.
 We now use the almost-Bogoliubov tranformation \eqref{eq:bogtrafo} (attention to include the factor $\lambda$ to the $\Xi$) to transform the $b^*_{-k}$ operator, arriving at
\begin{align*}
  &  \int_0^1 \di \lambda \sum_{k\in \Zbbn}\Xi(k)\hbar^2 k\cdot  \langle T(\lambda)\Omega, c^*_k b^*_{-k} T(\lambda)\Omega\rangle \\
  & =  \int_0^1 \di \lambda \sum_{k\in \Zbbn}\Xi(k) \hbar^2 k\cdot \langle \Omega, T(\lambda)^* c^*_k T(\lambda) \\
  & \hspace{12em} \times \left( \cc{\cosh(\lambda\Xi)(-k)}b^*_{-k} + \cc{\sinh(\lambda\Xi)(-k)} b_k + \mathcal{E}^*_{-k}(\lambda\Xi)\right)\Omega\rangle \\
  & =  \int_0^1 \di \lambda \sum_{k\in \Zbbn}\Xi(k) \hbar^2 k\cdot \langle \Omega, T(\lambda)^* c^*_k T(\lambda) \left( \cc{\cosh(\lambda\Xi)(-k)}b^*_{-k} +  \mathcal{E}^*_{-k}(\lambda\Xi)\right)\Omega\rangle \tagg{kintermone}
\end{align*}
The key observation is that by the Duhamel formula we obtain a commutator of $c^*_k$ with $B$, and this commutator can be expressed purely in terms of $b^*$--operators, for which we know the almost-Bogoliubov transformation rule. In detail,
\begin{equation}\label{eq:ctransform}
\begin{split}
  & T^*(\lambda) c^*_k T(\lambda) = c^*_k + \int_0^\lambda \di \mu\, T^*(\mu) [c^*_k,B] T(\mu) \\
  & = c^*_k + \int_0^\lambda \di \mu\, T^*(\mu)  \left(-\frac{1}{2} \sum_{l \in \Zbbn} \cc{\Xi(l)}\left( [c^*_k,b_{-l}]b_l + b_{-l}[c^*_k,b_l] \right)\right) T(\mu) \;.
\end{split}
\end{equation}
Similar to the commutator $[b^*_k,b_l] = -\delta_{k,l} + \mathcal{O}(N^{-2/3})$, also the commutator  $[c^*_k,b_l]$ can be thought of as a leading Kronecker delta (times a constant) and two other terms to be treated as errors; more precisely from Lemma~\ref{lem:kingkong} we have
\[
  [c^*_k,b_l] = - \delta_{k,l} f(k) + \mathcal{E}_\text{c}(k,l) \;.
\]
Let us plug this commutator into \eqref{eq:ctransform} to get (we use again $\Xi(k) = \Xi(-k)$)
\begin{align*}
  T^*(\lambda) c^*_k T(\lambda) & = c^*_k + \cc{\Xi(k)} f(k) \int_0^\lambda \di \mu  T^*(\mu) b_{-k} T(\mu) \\
  & \quad -\frac{1}{2} \sum_{l \in \Zbbn} \cc{\Xi(l)} \int_0^\lambda \di \mu T^*(\mu) \Big( \mathcal{E}_\text{c}(k,-l)b_l + b_{-l} \mathcal{E}_\text{c}(k,l)  \Big) T(\mu) \\
  & =: c^*_k + \cc{\Xi(k)} f(k) \int_0^\lambda \di \mu  T^*(\mu) b_{-k} T(\mu) + \mathcal{E}_\text{kin}(\Xi)(k) \;. \tagg{ctransfinal}
\end{align*}
The terms $\mathcal{E}_\text{c}$ and $\mathcal{E}_\text{kin}$ are 3-tuples of operators.
We plug \eqref{eq:ctransfinal} into \eqref{eq:kintermone}, use that $c_k$ acting on the vacuum vanishes, and then the almost-Bogoliubov transformation rule to find
\begin{align*}
  & \int_0^1 \di \lambda \sum_{k\in \Zbbn}\Xi(k)\hbar^2 k\cdot  \langle T(\lambda)\Omega, c^*_k b^*_{-k} T(\lambda)\Omega\rangle \\
  & = \int_0^1 \! \di \lambda \!\!\! \sum_{k \in \Zbbn}\!\!\! \Xi(k) \hbar^2 k\cdot \int_0^\lambda \di \mu \cc{\Xi(k)} f(k)\langle \Omega,T^*(\mu) b_{-k}T(\mu) \left(\cc{\cosh(\lambda\Xi)(-k)} b^*_{-k} \new{ + \mathcal{E}_{-k}^*(\lambda \Xi)} \right) \Omega\rangle \\
  & \quad + \int_0^1 \di\lambda \sum_{k \in \Zbbn} \Xi(k) \hbar^2 k\cdot \langle \Omega, \mathcal{E}_\text{kin}(\Xi)(k) \left(\cc{\cosh(\lambda\Xi)(-k)} b^*_{-k} + \mathcal{E}_{-k}^*(\lambda \Xi) \right) \Omega \rangle \\
  & = \int_0^1 \!\! \di \lambda\!\!\! \sum_{k \in \Zbbn} \!\!\lvert \Xi(k)\rvert^2 \hbar^2 k\cdot f(k) \int_0^\lambda \di \mu \langle \Omega,\Big(\cosh(\mu\Xi)(-k)b_{-k} + \sinh(\mu\Xi)(-k) b^*_{k} + \mathcal{E}_{-k}(\mu\Xi)\Big) \\
  & \hspace{22em} \times \left(\cc{\cosh(\lambda\Xi)(-k)} b^*_{-k} \new{+ \mathcal{E}_{-k}^*(\lambda \Xi)} \right) \Omega\rangle \\
  & \quad + \int_0^1 \di\lambda \sum_{k \in \Zbbn} \Xi(k) \hbar^2 k\cdot \langle \Omega, \mathcal{E}_\text{kin}(\Xi)(k) \left(\cc{\cosh(\lambda\Xi)(-k)} b^*_{-k} + \mathcal{E}_{-k}^*(\lambda \Xi) \right) \Omega \rangle \\
  & = \int_0^1 \di \lambda \sum_{k \in \Zbbn} \lvert \Xi(k)\rvert^2 \hbar^2 k\cdot f(k) \int_0^\lambda \di \mu \langle \Omega, \cosh(\mu\Xi)(-k)b_{-k} \cc{\cosh(\lambda\Xi)(-k)} b^*_{-k} \Omega\rangle + \varepsilon_2 \;,
\end{align*}
where we introduced the notation for the error term
\[
\begin{split}
  \varepsilon_2 & := \int_0^1 \di\lambda \sum_{k \in \Zbbn} \Xi(k) \hbar^2 k\cdot \langle \Omega, \mathcal{E}_\text{kin}(\Xi)(k) \left(\cc{\cosh(\lambda\Xi)(-k)} b^*_{-k} + \mathcal{E}_{-k}^*(\lambda \Xi) \right) \Omega \rangle\\
  & \quad + \int_0^1 \di \lambda \sum_{k \in \Zbbn} \lvert \Xi(k)\rvert^2 \hbar^2 k\cdot f(k) \int_0^\lambda \di \mu \langle \Omega, \mathcal{E}_{-k}(\mu\Xi) \left(\cc{\cosh(\lambda\Xi)(-k)} b^*_{-k} \new{ + \mathcal{E}_{-k}^*(\lambda \Xi) } \right) \Omega\rangle\;.
\end{split}
\]
To evaluate the scalar products in the main term, recall that $\langle \Omega, b_{-k} b^*_{-k}\Omega\rangle =1$, so
 \begin{align*}
 & \int_0^1 \di \lambda \sum_{k\in \Zbbn}\hbar^2 k\cdot \Xi(k) \langle T(\lambda)\Omega, c^*_k b^*_{-k} T(\lambda)\Omega\rangle\\
 & = \int_0^1 \di \lambda \sum_{k \in \Zbbn} \lvert \Xi(k)\rvert^2 \hbar^2 k\cdot f(k) \int_0^\lambda \di \mu \cosh(\mu\Xi)(-k) \cc{\cosh(\lambda\Xi)(-k)}  + \varepsilon_2 \\
 & = \int_0^1 \di \lambda \sum_{k \in \Zbbn} \lvert \Xi(k)\rvert^2 \hbar^2 k\cdot f(k) \frac{\sinh(\lambda \lvert \Xi(k)\rvert)}{\lvert \Xi(k)\rvert} {\cosh(\lambda\lvert \Xi(k)\rvert)}  + \varepsilon_2\\
 & = \sum_{k \in \Zbbn} \lvert \Xi(k)\rvert^2 \hbar^2 k\cdot f(k) \frac{\sinh(\lvert \Xi(k)\rvert)^2}{2\lvert \Xi(k)\rvert^2}  + \varepsilon_2\\
 & = \frac{1}{2}\sum_{k \in \Zbbn}  \hbar^2 k\cdot f(k) {\sinh(\lvert \Xi(k)\rvert)^2}  + \varepsilon_2 \;.
 \end{align*}
 The integrals over $\mu$ and $\lambda$ were computed explicitly.
 Plugging the last expression into \eqref{eq:dontforgetcc} (remember that a factor of $2$ arises from the $+\pluscc$), we obtain the statement of the lemma.
 
 \paragraph{Estimating the Error Terms.}
We use \eqref{eq:bstark} and \eqref{eq:bogtrafoerror} to estimate the second line of $\varepsilon_2$. We have two contributions to this line, due to splitting the last parenthesis containing $\cc{\cosh(\lambda\Xi)(-k)} b^*_{-k}$ and $\mathcal{E}_{-k}^*(\lambda \Xi)$. The first contribution to the second line is
\begin{align*}
  & \Bigg\lvert \int_0^1 \di \lambda \sum_{k \in \Zbbn} \lvert \Xi(k)\rvert^2 \hbar^2 k\cdot f(k) \int_0^\lambda \di \mu \langle \Omega,\mathcal{E}_{-k}(\mu\Xi) \cc{\cosh(\lambda\Xi)(-k)} b^*_{-k} \Omega\rangle \Bigg\rvert \\
  & \leq \int_0^1 \di\lambda \sum_{k \in \Zbb^3\setminus\{0\}}  \lvert \Xi(k)\rvert^2 \hbar^2 \lvert k\cdot f(k)\rvert \cosh(\lambda \lvert\Xi(k)\rvert)\int_0^\lambda \di \mu  \norm{\mathcal{E}^*_{-k}(\mu\Xi)\Omega} \norm{b^*_{-k}\Omega} \\
  & \leq 4\sum_{k \in \Zbbn} \lvert \Xi(k)\rvert 2 \hbar^2 \lvert k\cdot f(k)\rvert \sinh(\lvert \Xi(k)\rvert) \sup_{t \in [0,1]} \norm{(\Ncal+2)^{3/2}T(t)\Omega} e^{\lvert \Xi(k)\rvert} \sum_{m \in \Zbbn}  \frac{\lvert \Xi(m)\rvert}{n_m n_k} \;.
\end{align*}
Here we used that the integral over $\mu$ can be bounded by $1$ and the integral over $\lambda$ then turns $\cosh(\lambda \lvert \Xi(k)\rvert)$ into $\sinh(\lvert \Xi(k)\rvert)/\lvert \Xi(k)\rvert$.
\new{
The second contribution to the second line is bounded likewise, leading to
\begin{align*}
  & \Bigg\lvert \int_0^1 \di \lambda \sum_{k \in \Zbbn} \lvert \Xi(k)\rvert^2 \hbar^2 k\cdot f(k) \int_0^\lambda \di \mu \langle \Omega,\mathcal{E}_{-k}(\mu\Xi) \mathcal{E}_{-k}^*(\lambda \Xi) \Omega\rangle \Bigg\rvert \\
  & \leq \int_0^1 \di\lambda \sum_{k \in \Zbb^3\setminus\{0\}}  \lvert \Xi(k)\rvert^2 \hbar^2 \lvert k\cdot f(k)\rvert \int_0^\lambda \di \mu  \norm{\mathcal{E}^*_{-k}(\mu\Xi)\Omega} \norm{\mathcal{E}_{-k}^*(\lambda \Xi)\Omega} \\
  & \leq 16 \sum_{k \in \Zbbn} \hbar^2 \lvert k\cdot f(k)\rvert \sup_{t \in [0,1]} \langle T(t) \Omega, (\Ncal+2)^{3}T(t)\Omega\rangle e^{2\lvert \Xi(k)\rvert} \Bigg( \sum_{m \in \Zbbn}  \frac{\lvert \Xi(m)\rvert}{n_m n_k} \Bigg)^2 \;.
\end{align*}
}
We now consider the first line of $\varepsilon_2$. First of all,
we estimate $\norm{k\cdot\mathcal{E}^*_\text{kin}(\Xi)(k)\Omega}$. Using Lemma \ref{lem:kingkong} and \eqref{eq:bstark}, as well as the fact that $\mathcal{E}_\text{c}(k,l)$ commutes with $\Ncal$, we find
\[
\begin{split}
  & \norm{k\cdot\mathcal{E}^*_\text{kin}(\Xi)(k)\Omega} \\
  & = \frac{1}{2} \sum_{l \in \Zbbn} \lvert \Xi(l)\rvert \int_0^\lambda \di \mu \norm{\Big( k\cdot b^*_l \mathcal{E}^*_\text{c}(k,-l) +  k\cdot \mathcal{E}^*_\text{c}(k,l) b_{-l}  \Big) T(\mu) \Omega} \\
  & \leq \!\!\!\sum_{l \in \Zbbn} \!\!\!\lvert \Xi(l)\rvert \! \int_0^\lambda \!\!\! \di \mu \Big( \norm{(\Ncal+1)^{1/2} k\cdot \mathcal{E}_\text{c}(k,l)  T(\mu) \Omega} + \new{\lvert k\rvert \frac{2 \Big(\left(\frac{3}{4\pi} \right)^{1/3} N^{1/3} + \lvert k\rvert \Big)}{n_k n_l}}\norm{ \Ncal b_l T(\mu)\Omega} \Big)\\
  & \leq \!\!\!\sum_{l \in \Zbbn} \!\!\!\lvert \Xi(l)\rvert \! \int_0^\lambda \!\!\! \di \mu \Big( \norm{k\cdot \mathcal{E}_\text{c}(k,l)(\Ncal+1)^{1/2}   T(\mu) \Omega} +  \new{\lvert k\rvert \frac{2 \Big(\left(\frac{3}{4\pi} \right)^{1/3} N^{1/3} + \lvert k\rvert \Big)}{n_k n_l}}\norm{ b_l \Ncal  T(\mu)\Omega}  \Big) \\
  & \leq \sum_{l \in \Zbbn} \lvert \Xi(l)\rvert \new{ 8 \lvert k\rvert \frac{4 \Big(\left(\frac{3}{4\pi} \right)^{1/3} N^{1/3} + \lvert k\rvert \Big)}{n_k n_l}} \sup_{\mu \in [0,1]} \norm{ (\Ncal+1)^{3/2}  T(\mu)\Omega} \;.
\end{split}
\]
Using this estimate together with $\langle \Omega, b_{-k} b^*_{-k}\Omega\rangle =1$ and \eqref{eq:bogtrafoerror} we find that the entire first line of $\varepsilon_2$ can be estimated to be
\begin{align*}
  & \left\lvert \int_0^1 \di\lambda \sum_{k \in \Zbbn} \Xi(k)\hbar^2 k\cdot \langle \Omega, \mathcal{E}_\text{kin}(\Xi)(k) \left(\cc{\cosh(\lambda\Xi)(-k)} b^*_{-k} + \mathcal{E}_{-k}^*(\lambda \Xi) \right) \Omega \rangle \right\rvert \\
  & \leq \int_0^1 \di\lambda \sum_{k \in \Zbbn} \lvert \Xi(k)\rvert \hbar^2 \norm{k\cdot\mathcal{E}^*_\text{kin}(\Xi)(k)\Omega} \Big( \cosh(\lambda \lvert \Xi(k)\rvert) \norm{b^*_{-k}\Omega} + \norm{\mathcal{E}^*_{-k}(\lambda\Xi) \Omega} \Big) \\
  & \leq \int_0^1 \di\lambda \sum_{k \in \Zbbn} \lvert \Xi(k)\rvert 2 \hbar^2 \norm{k\cdot\mathcal{E}^*_\text{kin}(\Xi)(k)\Omega} \\
  & \hspace{2cm} \times \Big(  \cosh(\lambda \lvert \Xi(k)\rvert) + 2\sup_{t \in [0,1]} \norm{(\Ncal+2)^{3/2}T(t)\Omega} e^{\lambda\lvert \Xi(k)\rvert} \sum_{m \in \Zbbn}  \frac{\lambda \lvert \Xi(m)\rvert}{n_m n_k}\Big) \\
  & \leq \int_0^1 \di\lambda \sum_{k \in \Zbbn} \lvert \Xi(k)\rvert 4 \hbar^2 \sum_{l \in \Zbbn} \lvert \Xi(l)\rvert   \new{\lvert k\rvert \frac{4 \Big(\left(\frac{3}{4\pi} \right)^{1/3} N^{1/3} + \lvert k\rvert \Big)}{n_k n_l}} \sup_{\mu \in [0,1]} \norm{ (\Ncal+1)^{3/2}  T(\mu)\Omega} \\
  & \hspace{2cm} \times \Big(  \cosh(\lambda \lvert \Xi(k)\rvert) + 2\sup_{t \in [0,1]} \norm{(\Ncal+2)^{3/2}T(t)\Omega} e^{\lambda\lvert \Xi(k)\rvert} \sum_{m \in \Zbbn}  \frac{\lambda \lvert \Xi(m)\rvert}{n_m n_k}\Big) \\
  & \leq \sum_{k \in \Zbbn} 4 \hbar^2 \sum_{l \in \Zbbn} \lvert \Xi(l)\rvert   \new{\lvert k\rvert \frac{4 \Big(\left(\frac{3}{4\pi} \right)^{1/3} N^{1/3} + \lvert k\rvert \Big)}{n_k n_l}} \sup_{\mu \in [0,1]} \norm{ (\Ncal+1)^{3/2}  T(\mu)\Omega} \\
  & \hspace{2cm} \times \Big(  \sinh(\lvert \Xi(k)\rvert) + 2\sup_{t \in [0,1]} \norm{(\Ncal+2)^{3/2}T(t)\Omega} e^{\lvert \Xi(k)\rvert} \sum_{m \in \Zbbn}  \frac{ \lvert \Xi(m)\rvert}{n_m n_k}\Big) \;.
\end{align*}
To get from the second last to the last line, we have estimated the $\lambda$ in the $m$-sum by $1$ and then integrated the $\cosh$ and the exponential explicitly.
\end{proof}

\section{Optimizing the Trial State}
We optimize the choice of the almost-Bogoliubov transform with respect to $\Xi$. We are later going to calculate $n_k^2$ and $k\cdot f(k)$ explicitly, and then we will see that both $\alpha_k$ and $\beta_k$, and therefore the whole energy correction, are of order $\hbar$. This also sets the scale for our error bounds: we have to show that all errors are at least as small as $o(\hbar)$.

\begin{prp}[Minimal almost-Bogoliubov Energy]\label{lem:minbogenergy}
Let $\hat{V}(k) > 0$. The lowest energy among almost-Bogoliubov trial states is found by taking for all $k$ the function $\Xi_0$ to be
\[
  \Xi_0(k) = -\frac{1}{2}\artanh\left(\frac{\beta_k}{\alpha_k}\right) = -\frac{1}{4} \log\left( \frac{1+\frac{\beta_k}{\alpha_k}}{1-\frac{\beta_k}{\alpha_k}} \right) \;.
\]
The minimal value of the functional defined in \eqref{eq:functionaldef} below is
\begin{equation}\label{eq:minenergy}
  \inf_{\Xi} E_\text{bosonized}(\Xi) = E_\text{bosonized}(\Xi_0) = \sum_{k \in \Zbbn} \frac{1}{2} \left( \sqrt{\alpha_k^2 - \beta_k^2} - \alpha_k \right) < 0
\end{equation}
with the coefficients
\begin{equation} \label{eq:alphabeta}
  \alpha_k := \hbar^2 k\cdot f(k) + \frac{1}{N} \hat{V}(k) n_k^2 > 0\;, \qquad \beta_k := \frac{1}{N} \hat{V}(k) n_k^2 > 0 \;.
\end{equation}
Note that $k\cdot f(k) > 0$, so $\alpha_k > \beta_k$.
\end{prp}

\begin{proof}
From Lemma \ref{lem:interaction} and Lemma \ref{lem:kinetic} we have
\begin{align*}
  & \langle T \Omega, \left( \Hbb_0 + Q_N^{(0)} \right) T \Omega \rangle \\
  & = \frac{1}{N} \Re \sum_{k \in \Zbbn} \hat{V}(k) n_k^2 \left( \sinh(\lvert\Xi(k)\rvert)^2 + \frac{\cc{\Xi(k)}}{\lvert \Xi(k)\rvert} \sinh(\lvert \Xi(k)\rvert) \cosh(\lvert \Xi(k)\rvert)\right) \\
  & \quad + \sum_{k \in \Zbbn}  \hbar^2 k\cdot f(k) {\sinh(\lvert \Xi(k)\rvert)^2}  + \varepsilon_1 + 2 \Re \varepsilon_2\\
  & =: E_\text{bosonized}(\Xi) +  \varepsilon_1 + 2 \Re \varepsilon_2 \;.  \tagg{functionaldef}
\end{align*}
We minimize the functional $E_\text{bosonized}(\Xi)$. Due to the real part on the first line, the complex part of $\Xi(k)$ does not matter, so w.\,l.\,o.\,g.\ $\Xi(k) \in \Rbb$. We can optimize for every $k$ separately. The coefficients $\frac{1}{N}\hat{V}(k) n_k^2$ and $\frac{\hbar^2}{2} k \cdot f(k)$ are both non-negative, and also $\sinh(\lvert \Xi(k)\rvert)^2$ and $\sinh(\lvert \Xi(k)\rvert)\cosh(\lvert \Xi(k)\rvert)$ are also non-negative. So the only way of obtaining a result smaller than zero is obtained for $\cc{\Xi(k)}/\lvert \Xi(k)\rvert = -1$, that is, for $\Xi(k)$ a negative function.
We have
\[
  E_\text{bosonized}(\Xi) = \sum_{k \in \Zbbn} g_k(\lvert \Xi(k)\rvert) \;,
\qquad
  g_k(x) := \alpha_k \sinh(x)^2 - \beta_k \sinh(x) \cosh(x)
\]
with the coefficients $\alpha_k$ and $\beta_k$ as given above.
Therefore we have to minimize $g_k(x)$ with respect to $x \geq 0$, with $k$ as a parameter. We determine the critical points:
\begin{align*}
  0 \overset{!}{=} g'_k(x_k) & = \alpha_k 2 \sinh(x_k) \cosh(x_k) - \beta_k \left( \cosh(x_k)^2 + \sinh(x_k)^2 \right) \\
  & = \alpha_k \sinh(2x_k) - \beta_k \cosh(2x_k) \;.
\end{align*}
This has the positive (since $0 < \frac{\beta_k}{\alpha_k} < 1$) solution
\begin{align*}
  2x_k = \artanh\left( \frac{\beta_k}{\alpha_k} \right) \;.
\end{align*}
We plug this into $g_k(x) = \alpha_k \frac{1}{2}\left( \cosh(2x)-1 \right) - \beta_k \frac{1}{2} \sinh(2x)$ and use the two identities $\cosh(\artanh(A)) = 1/\sqrt{1-A^2}$ and $\sinh(\artanh(A)) = A/\sqrt{1-A^2}$ to obtain
\[
  g_k(x_k) = \frac{1}{2} \left(\alpha_k \frac{1}{\sqrt{1-(\beta_k/\alpha_k)^2}} - \alpha_k - \beta_k \frac{\beta_k/\alpha_k}{\sqrt{1-(\beta_k/\alpha_k)^2}} \right) = \frac{1}{2} \left( \sqrt{\alpha_k^2 - \beta_k^2} - \alpha_k \right) \;.
\]
This confirms \eqref{eq:minenergy}.
\end{proof}

\begin{lem}[Constant for the Kinetic Energy]
\label{lem:constants}
The function $f$ defined above satisfies
\begin{equation}\label{eq:constantfk}
  k\cdot f(k) = \lvert k\rvert N^{1/3} \left( \frac{4}{3\sqrt{\pi}} \right)^{2/3}  + \Ocal(N^{1/6})\;.
\end{equation}
Moreover
\[
  \beta_k = \hbar \left( \frac{3}{4}\sqrt{\pi} \right)^{2/3} \hat{V}(k) \lvert k\rvert + \Ocal(N^{-1/2})\;, \qquad \alpha_k = \hbar \lvert k\rvert \left( \frac{4}{3\sqrt{\pi}} \right)^{2/3}+ \beta_k + \Ocal(N^{-1/2})\;.
\]
\end{lem}
\begin{proof}
We have
\[
  n_k^2 k\cdot f(k) = \sum_{\substack{p\in B_\textnormal{F}^c\\h\in B_\textnormal{F}}} \delta_{p-h,k} k\cdot (2h+k) \simeq 2 \sum_{h \in B_\textnormal{F}}\chi(\lvert h+k\rvert > R) k\cdot h \;,
\]
where $R = \left(\frac{3}{4\pi}\right)^{1/3} N^{1/3}$ is the leading order of the Fermi momentum. Now we use an integral approximation. Rescaling in the sum $h = N^{1/3} \tilde h$, where $\tilde h$ now corresponds to Riemann cubes of side length $N^{-1/3}$, the indicated errors are one power of $N^{-1/3}$ smaller than the main terms. We use the symbol $\simeq$ where these error terms have been suppressed.)
\begin{align*}
  & 2 \sum_{h \in B_\textnormal{F}}\chi(\lvert h+k\rvert > R) k\cdot h\\
  & \simeq 2N \int_{\lvert \tilde{h}\rvert \leq \left(\frac{3}{4\pi}\right)^{1/3}} \di^3 \tilde{h}\ \chi\left(\lvert \tilde h + \frac{k}{N^{1/3}}\rvert > \left(\frac{3}{4\pi}\right)^{1/3}\right) k\cdot \left(N^{1/3}\tilde h\right) \tagg{step0}\\
  & = 2 N^{4/3} \int_0^{\left(\frac{3}{4\pi}\right)^{1/3}} r^2 \di r \int_0^{\frac{\pi}{2}+\text{small}} \sin(\theta) \di \theta \int_0^{2\pi} \di \varphi\,\chi\left( \lvert \tilde h + \frac{k}{N^{1/3}} \rvert > \left( \frac{3}{4\pi} \right)^{1/3} \right) r \cos(\theta)\lvert k\rvert \tagg{step1}\\
  & \simeq 2 (2\pi) \lvert k\rvert N^{4/3}  \int_0^{\frac{\pi}{2}} \sin(\theta) \di \theta \cos(\theta) \int_{\left(\frac{3}{4\pi}\right)^{1/3} - \lvert k\rvert \cos(\theta) N^{-1/3}}^{\left(\frac{3}{4\pi}\right)^{1/3}} r^3 \di r \tagg{step2}\\
  & = 2(2\pi) \lvert k\rvert N^{4/3} \int_0^{\frac{\pi}{2}} \di\theta \sin(\theta)\cos(\theta) \frac{1}{4} \left( \left(\frac{3}{4\pi}\right)^{4/3} - \left(\left(\frac{3}{4\pi}\right)^{1/3} - \frac{\lvert k\rvert \cos(\theta)}{N^{1/3}}\right)^4 \right) \tagg{step3}\\
  & = 2(2\pi) \lvert k\rvert N^{4/3}  \left( \frac{1}{4\pi} \lvert k\rvert N^{-1/3} + \mathcal{O}(N^{-2/3}) \right) \tagg{step4}\\
  & \simeq \lvert k\rvert^2 N \;. \tagg{laststep}
\end{align*}
To get from \eqref{eq:step0} to \eqref{eq:step1}, we parametrized the vector $\tilde h$ in spherical coordinates by its length $r$, the angle $\theta$ measured between $\tilde h$ and $k$, and the remaining rotation by the angle $\varphi$.
To get from \eqref{eq:step1} to \eqref{eq:step2}, we wrote the condition of the characteristic function as
\[
  r^2  + 2r \frac{\lvert k\rvert}{N^{1/3}} \cos(\theta) + \frac{\lvert k\rvert^2}{N^{2/3}}> \left( \frac{3}{4\pi} \right)^{2/3} \;,
\]
which is satisfied for
\[
  r \geq \sqrt{\left( \frac{3}{4\pi} \right)^{2/3} - \frac{\lvert k\rvert^2}{N^{2/3}} (\sin \theta)^2} - \frac{\lvert k\rvert}{N^{1/3}} \cos(\theta) \simeq \left( \frac{3}{4\pi} \right)^{1/3} - \frac{\lvert k\rvert}{N^{1/3}} \cos(\theta) \;.
\]
Moreover from \eqref{eq:step3} to \eqref{eq:step4} we used the asymptotic relation $a^4 - (a-x)^4 = 4 a^3 x + \mathcal{O}(x^2)$ as $x \to 0$, given $a \in \Rbb$.
Dividing \eqref{eq:laststep} by $n_k^2 = \lvert k\rvert N \hbar \left( \frac{3}{4} \sqrt{\pi}\right)^{2/3} + \mathcal{O}(N^{1/2})$ we obtain the claimed formula for $k\cdot f(k)$. Using this result in \eqref{eq:alphabeta} we obtain the claims for $\alpha_k$ and $\beta_k$.
\end{proof}

\section{Estimating the Error Terms}
We have estimated all error terms through the expectation value
\[
  \sup_{t \in [0,1]} \langle T(t)\Omega,(\Ncal+2)^3 T(t) \Omega \rangle
\]
of the fermionic number operator $\Ncal$. We estimate this expectation value using Gr\"onwall's lemma as in the context of the time-dependent problem \cite{BPS, BPS2, BPSetal, coulomb,BBMN25} and of the momentum distribution \cite{BL25,BLN25}.

\begin{prp}[Bound on the Fermionic Particle Number]\label{lem:particlenumber}
For all $n \in \Nbb$ and for all $\psi \in \fock$ we have 
\[
  \sup_{t \in [0,1]} \langle T(t)\psi, (\Ncal+1)^n T(t) \psi \rangle \leq e^{C_n(\Xi)} \langle \psi, (\Ncal+1)^n \psi \rangle \;,
\]
where the constant is
\[
  C_n(\Xi) := 8 n (5^n)\sum_{k \in \Zbbn} \lvert \Xi(k)\rvert \;.
\]
In particular, for fixed $n$ and in the vacuum $\psi = \Omega$, the bound is of order $1$.
\end{prp}
\begin{proof}
Using the CAR for the fermionic operators we obtain
\begin{equation}\label{eq:commutators}
  [\Ncal,b^*_k] = 2b^*_k \quad \text{and} \quad b^*_k b^*_{-k} (\Ncal+4) = \Ncal b^*_k b^*_{-k} \;.
\end{equation}
We calculate the derivative of the expectation value:
\begin{align*}
  & \left\lvert \frac{\di}{\di t} \langle T(t)\psi, (\Ncal+1)^n T(t)\psi \rangle \right \rvert \\
  & = \left\lvert \langle T(t)\psi, \sum_{j=0}^{n-1} (\Ncal+1)^j [\Ncal,B](\Ncal+1)^{n-j-1} T(t)\psi \rangle \right \rvert \\
  & = \left \lvert 2 \sum_{k \in \Zbbn} \Xi(k) \sum_{j=0}^{n-1} \langle T(t)\psi, (\Ncal+1)^j b^*_k b^*_{-k} (\Ncal+1)^{n-j-1} T(t)\psi\rangle + \pluscc \right\rvert;
  \intertext{we now insert $\id = (\Ncal+5)^{\frac{n-1}{2}-j} (\Ncal+5)^{j-\frac{n-1}{2}}$ after $b^*_k b^*_{-k}$ and commute in order to distribute the powers of the number operator equally onto both arguments of the scalar product:}
  & = \left \lvert 4 \Re\!\!\!\!\!\! \sum_{k \in \Zbbn}\!\!\!\!\! \Xi(k) \sum_{j=0}^{n-1} \langle T(t)\psi, (\Ncal+1)^j (\Ncal+1)^{\frac{n-1}{2}-j} b^*_k b^*_{-k} (\Ncal+5)^{j-\frac{n-1}{2}}(\Ncal+1)^{n-j-1} T(t)\psi\rangle \right\rvert \\
  & \leq 4\!\!\!\! \sum_{k \in \Zbbn}\!\! \lvert \Xi(k)\rvert \sum_{j=0}^{n-1} \norm{b_k (\Ncal+1)^{\frac{n-1}{2}-j} (\Ncal+1)^j T(t)\psi} \norm{b^*_{-k} (\Ncal+5)^{j-\frac{n-1}{2}}(\Ncal+1)^{n-j-1} T(t)\psi} \\
  & \leq 8\sum_{k \in \Zbbn} \lvert \Xi(k)\rvert \sum_{j=0}^{n-1} \norm{\Ncal^{1/2} (\Ncal+1)^{\frac{n-1}{2}-j} (\Ncal+1)^j T(t)\psi}  \norm{(\Ncal+1)^{1/2} (\Ncal+5)^{\frac{n}{2}-\frac{1}{2}} T(t)\psi} \\
  & \leq 8\sum_{k \in \Zbbn} \lvert \Xi(k)\rvert \ n\, \langle T(t)\psi, (\Ncal+5)^n T(t)\psi\rangle \\
  & \leq C_n(\Xi)\langle T(t)\psi, (\Ncal+1)^n T(t)\psi\rangle  \;.
\end{align*}
The claim now follows through Gr\"onwall's Lemma.
\end{proof}

\begin{lem}[Collected Error Estimates]\label{lem:collectederrors}
Let $c := 4 \left(\frac{9 \pi}{16} \right)^{2/3}$. Assume that the following quantities are finite:
\begin{align*}
  A_1 & := \sum_{k \in \Zbbn} \left\lvert \log\left( 1+c\hat{V}(k) \right) \right\rvert  \sqrt{1+c\hat{V}(k)}\sqrt{\lvert k\rvert} \;, &  A_2& := \sum_{k \in \Zbbn} \lvert \hat{V}(k)\rvert \sqrt{1+c\hat{V}(k)} \sqrt{\lvert k\rvert} \;, \\
  A_3 &:= \sum_{k \in \Zbbn} \sqrt{1+c\hat{V}(k)} \lvert k\rvert^{3/2} \;.
\end{align*}
Then there exists $C>0$, depending only on $A_1$ through $A_3$, such that the sum of all error terms is
\[
  \lvert \varepsilon_1 + 2 \Re \varepsilon_2 \rvert \leq \frac{C}{N}  \;.
\]
\end{lem}
\begin{proof}
We consider the Bogoliubov transformation found in Proposition \ref{lem:minbogenergy}, defined by
\[
  \Xi_0(k) = -\frac{1}{2}\artanh\left(\frac{\beta_k}{\alpha_k}\right) = -\frac{1}{4} \log\left( \frac{1+\frac{\beta_k}{\alpha_k}}{1-\frac{\beta_k}{\alpha_k}} \right) \;,
\]
with $\beta_k$, $\alpha_k$, and $\lvert k \cdot f(k)\rvert$ as computed in Lemma~\ref{lem:constants}. Using Proposition \ref{lem:particlenumber} we get
\[
  \sup_{t \in [0,1]} \langle T(t)\Omega, (\Ncal+2)^3 T(t)\Omega \rangle \leq 8 e^{C_3(\Xi_0)} \;.
\]
Using this bound in the estimates from Proposition~\ref{lem:interaction} and Proposition~\ref{lem:kinetic} we get
\begin{align}
  \lvert \varepsilon_1 \rvert & \leq \bigg\lvert \frac{1}{N} \sum_{k \in \Zbbn} \hat{V}(k) n_k^2 \Bigg[ 32 e^{C_3(\Xi_0)} e^{2\lvert \Xi_0(k)\rvert} \Bigg(\sum_{m \in \Zbb^3 \setminus\{0\}} \frac{\lvert \Xi_0(m)\rvert}{n_m n_k}\Bigg)^2  \label{eq:epsoneone}\\
  &\quad +  \big(4\sinh(\lvert \Xi_0(k)\rvert)+2 \cosh(\lvert \Xi_0(k)\rvert) \big) \sqrt{8} e^{C_3(\Xi_0)/2} e^{\lvert \Xi_0(k)\rvert} \sum_{m \in \Zbbn}  \frac{\lvert \Xi_0(m)\rvert}{n_m n_k} \Bigg] \bigg\rvert\label{eq:epsonetwo}
\end{align}
and
\begin{align}
  \lvert \varepsilon_2 \rvert & \leq  \hbar^2 64 e^{C_3(\Xi_0)} \bigg\lvert\sum_{k \in \Zbbn} \Bigg[ 2\lvert \Xi_0(k)\rvert \lvert k\cdot f(k)\rvert \sinh(\lvert \Xi_0(k)\rvert)  e^{\lvert \Xi_0(k)\rvert} \sum_{m \in \Zbbn}  \frac{2\lvert \Xi_0(m)\rvert}{n_m n_k} \label{eq:epstwoone}\\
  &  \quad + \sum_{l \in \Zbbn} \lvert \Xi_0(l)\rvert  \new{\lvert k\rvert\frac{8 \Big(\left(\frac{3}{4\pi} \right)^{1/3} N^{1/3} + \lvert k\rvert \Big)}{n_k n_l}} \Big(  \sinh(\lvert \Xi_0(k)\rvert) +  e^{\lvert \Xi_0(k)\rvert} \sum_{m \in \Zbbn}  \frac{ \lvert \Xi_0(m)\rvert}{n_m n_k}\Big) \label{eq:epstwotwo}\\
  & \quad \new{+ \lvert k\cdot f(k)\rvert e^{2\lvert \Xi_0(k)\rvert} \Bigg( \sum_{m \in \Zbbn}  \frac{\lvert \Xi_0(m)\rvert}{n_m n_k} \Bigg)^2} \Bigg] \bigg\rvert\;.  \label{eq:epstwothree}
\end{align}
Let
\[
   A_4 := \sum_{k \in \Zbbn} \left\lvert \log\left( 1+ c\hat{V}(k) \right) \right\rvert \;.
\]
This is finite since $A_4 \leq C A_1$ for some constant $C >0$ depending only on $\hat{V}$.
First observe that we have
  $\Xi_0(k) = - \frac{1}{4} \log \big( 1+ c \hat{V}(k) \big) + \mathcal{O}(N^{-1/3})$,
by which we find that
$
  C_3(\Xi_0) \simeq 750 \sum_{k \in \Zbbn} \big\lvert \log\big( 1+ c\hat{V}(k) \big) \big\rvert = 750 A_4
  $ and
$e^{2\lvert \Xi_0(k)\rvert} \simeq (1+c \hat{V}(k))^{1/2}$, where the lower order terms are not of relevance in the following as we are just looking for an upper bound.
We estimate the first line of $\varepsilon_1$, as given in \eqref{eq:epsoneone}, by
\begin{align*}
  & \bigg\lvert \frac{1}{N} \sum_{k \in \Zbbn} \hat{V}(k) n_k^2 32 e^{C_3(\Xi_0)} e^{2\lvert \Xi_0(k)\rvert} \Bigg(\sum_{m \in \Zbb^3 \setminus\{0\}} \frac{\lvert \Xi_0(m)\rvert}{n_m n_k}\Bigg)^2 \bigg\rvert \\
  & \leq \frac{1}{N^{5/3}} \left( \frac{4}{3\sqrt{\pi}} \right)^{2/3} 32 e^{750 A_4} \sum_{k \in \Zbbn} \lvert \hat{V}(k)\rvert \sqrt{1+c \hat{V}(k)} \Bigg(\sum_{m \in \Zbb^3 \setminus\{0\}} \frac{\lvert \Xi_0(m)\rvert}{\sqrt{\lvert m \rvert}}\Bigg)^2\\
  & \leq \frac{1}{N^{5/3}} \left( \frac{4}{3\sqrt{\pi}} \right)^{2/3} 2 e^{750 A_4} A_2 A_4^2 \;,
\end{align*}
where we used that, because of $\lvert m\rvert \geq 1$,
$\sum_{m \in \Zbb^3 \setminus\{0\}} \frac{\lvert \Xi_0(m)\rvert}{\sqrt{\lvert m \rvert}} \leq \frac{1}{4} A_4$.

For the second line of $\varepsilon_1$, as given in \eqref{eq:epsonetwo}, recall that $\sinh(\lvert \Xi_0(k)\rvert) \leq e^{\lvert \Xi_0(k)\rvert}$ and $\cosh(\lvert \Xi_0(k)\rvert) \leq e^{\lvert \Xi_0(k)\rvert}$. Using these two estimates and then proceeding as above, we find
\begin{align*}
  & \bigg\lvert \frac{1}{N} \sum_{k \in \Zbbn} \hat{V}(k) n_k^2 \big(4\sinh(\lvert \Xi_0(k)\rvert)+2 \cosh(\lvert \Xi_0(k)\rvert) \big) \sqrt{8} e^{C_3(\Xi_0)/2} e^{\lvert \Xi_0(k)\rvert} \sum_{m \in \Zbbn}  \frac{\lvert \Xi_0(m)\rvert}{n_m n_k} \bigg\rvert \\
  & \leq \frac{1}{N} 6\sqrt{8} e^{375 A_4} \sum_{k \in \Zbbn} \lvert \hat{V}(k)\rvert \sqrt{\lvert k\rvert} \sqrt{1+c\hat{V}(k)}\ \frac{1}{4}A_4\\
  & = \frac{1}{N} 3 \sqrt{2} e^{375 A_4} A_2 A_4  \;.
\end{align*}
For the first line of $\varepsilon_2$, as given in \eqref{eq:epstwoone}, we have
\begin{align*}
  & \bigg\lvert \hbar^2 64 e^{C_3(\Xi_0)}\sum_{k \in \Zbbn} 4\lvert \Xi_0(k)\rvert \lvert k\cdot f(k)\rvert \sinh(\lvert \Xi_0(k)\rvert)  e^{\lvert \Xi_0(k)\rvert} \sum_{m \in \Zbbn}  \frac{\lvert \Xi_0(m)\rvert}{n_m n_k}\bigg\rvert \\
  & \leq \hbar^2 256 \left( \frac{4}{3\sqrt{\pi}} \right)^{2/3} e^{750 A_4} \sum_{k\in \Zbbn} \lvert \Xi_0(k)\rvert \lvert k\cdot f(k) \rvert e^{2\lvert \Xi_0(k)\rvert} \frac{1}{N^{2/3}\sqrt{\lvert k\rvert}} \sum_{m \in \Zbbn} \frac{\lvert \Xi_0(m)\rvert}{\sqrt{\lvert m\rvert}} \\
  & \leq \frac{16}{N^{4/3}} {e^{750 A_4}} \left( \frac{4}{3\sqrt{\pi}} \right)^{2/3}  \sum_{k \in \Zbbn} \left\lvert \log\left( 1+c\hat{V}(k) \right) \right\rvert \frac{\lvert k\cdot f(k)\rvert}{\sqrt{\lvert k \rvert}} \sqrt{1+c\hat{V}(k)}\ A_4 \\
  & \leq \frac{16}{N} {e^{750 A_4}} \left( \frac{4}{3\sqrt{\pi}} \right)^{4/3}  A_1 A_4  \;.
\end{align*}
To estimate the second line of $\varepsilon_2$, as given in \eqref{eq:epstwotwo}, we have
\begin{align*}
  & \bigg\lvert \hbar^2 64 e^{C_3(\Xi_0)} \!\!\!\!\!\sum_{k \in \Zbbn} \sum_{l \in \Zbbn} \lvert \Xi_0(l)\rvert  \new{\lvert k\rvert\frac{8 \Big(\left(\frac{3}{4\pi} \right)^{1/3} N^{1/3} + \lvert k\rvert \Big)}{n_k n_l}} \\
  & \hspace{14em} \times \Big(  \sinh(\lvert \Xi_0(k)\rvert) +  e^{\lvert \Xi_0(k)\rvert} \!\!\!\!\! \sum_{m \in \Zbbn}\!\!  \frac{ \lvert \Xi_0(m)\rvert}{n_m n_k}\Big) \bigg\rvert \\
  & \leq \frac{512 e^{750 A_4}}{N}\left( \frac{16}{9\pi} \right)^{1/3}   \sum_{l \in \Zbbn}  \frac{\lvert \Xi_0(l)\rvert}{\sqrt{\lvert l\rvert}}  \sum_{k \in \Zbbn} \Big(\Big( \frac{3}{4\pi} \Big)^{1/3} + \frac{\lvert k \rvert}{N^{1/3}} \Big) \sqrt{\lvert k\rvert}  \\
  & \quad \times e^{\lvert \Xi_0(k)\rvert} \Bigg(1+ \sum_{m \in \Zbbn} \frac{\lvert \Xi_0(m)\rvert}{\sqrt{\lvert m\rvert} }\frac{1}{N^{2/3}\sqrt{\lvert k \rvert}} \left( \frac{4}{3\sqrt{\pi}} \right)^{2/3} \Bigg)\\
  & \leq \frac{128 e^{750 A_4}}{N} \left( \frac{16}{9\pi} \right)^{1/3} A_4  A_3 \left(1+ \left( \frac{4}{3\sqrt{\pi}} \right)^{2/3}\frac{A_4}{N^{2/3}}  \right) \;.
\end{align*}
\new{
To estimate the third line of $\varepsilon_2$, as given in \eqref{eq:epstwothree}, using Lemma~\ref{lem:constants}, we have
\begin{align*}
   & \bigg\lvert \hbar^2 64 e^{C_3(\Xi_0)} \sum_{k \in \Zbbn} \lvert k\cdot f(k)\rvert e^{2\lvert \Xi_0(k)\rvert} \Bigg( \sum_{m \in \Zbbn}  \frac{\lvert \Xi_0(m)\rvert}{n_m n_k} \Bigg)^2 \bigg\rvert \\
   & \leq \hbar^2 \frac{1024}{9\pi} e^{750 A_4} \sum_{k \in \Zbbn} \lvert k\rvert N^{1/3}   e^{2\lvert \Xi_0(k)\rvert} \Bigg( \sum_{m \in \Zbbn}  \frac{\lvert \Xi_0(m)\rvert}{\sqrt{\lvert k\rvert} \sqrt{\lvert m \rvert} N \hbar}  \Bigg)^2 \\
   & \leq \frac{1}{N^{5/3}} \frac{64}{9\pi} e^{750 A_4} A_3 A_4^2 \;.
   \end{align*}
}
Combining these five estimates we find the claimed result.
\end{proof}

\section{Proof of the Main Result}
\begin{proof}[Proof of Theorem~\ref{thm:mainresult}]
From \eqref{eq:phtransform}, we have
\[
  \langle R_\omega T\Omega, H_N R_\omega T \Omega \rangle = \mathcal{E}_\text{HF}(\omega) + \langle T\Omega, \left( \di\Gamma(uhu-\cc{v}\cc{h}v) + Q_N\right)T\Omega\rangle \;.
\]
In \eqref{eq:splitinteraction} we splitted off two error operators $\mathcal{E}_1$ and $\mathcal{E}_2$,
\[
  Q_N = Q_N^{(0)} + \frac{1}{2N} \int_{\Tbb^3 \times \Tbb^3} \di x\di y V(x-y) \big(\mathcal{E}_1 + \mathcal{E}_2 \big) \;.
\]
By Lemma \ref{lem:errornoparthole} combined with Proposition \ref{lem:particlenumber}, we can estimate the first error operator by
\[
  \lvert \langle T \Omega, \frac{1}{2N}\int_{\Tbb^3\times \Tbb^3} \di x\di y\, V(x-y) \,\mathcal{E}_1\, T \Omega\rangle\rvert \leq \frac{2}{N} \sum_{k \in \Zbb^3} \lvert \hat{V}(k)\rvert e^{C_2(\Xi_0)} = \mathcal{O}(N^{-1})\;.
\]
Furthermore it is easy to check that $[B,i^\Ncal] = 0$, and thus $T i^\Ncal = i^\Ncal T$,
so by Lemma \ref{lem:errorparity} the second error operator's contribution to the expectation value vanishes:
\[
  \langle T \Omega, \frac{1}{2N}\int_{\Tbb^3\times \Tbb^3} \di x\di y\, V(x-y) \,\mathcal{E}_2\, T \Omega\rangle  = 0\;.
\]
From Lemma \ref{lem:freekinetic} combined with Proposition \ref{lem:particlenumber} we learn that direct and exchange term can be dropped, retaining only $\Hbb_0 := \di\Gamma\left(u(-\hbar^2\Delta)u - \cc{v}(-\hbar^2\Delta)v\right)$ for the kinetic energy:
\[
  \langle T\Omega, \di\Gamma(uhu-\cc{v}\cc{h}v) T\Omega \rangle = \langle T\Omega, \Hbb_0 T\Omega \rangle + \mathcal{O}(N^{-1}) \;.
\]
It remains
\[
  \langle R_\omega T\Omega, H_N R_\omega T \Omega \rangle = \mathcal{E}_\text{HF}(\omega) + \langle T\Omega, \left( \Hbb_0 + Q^{(0)}_N\right)T\Omega\rangle + \mathcal{O}(N^{-1})\;.
\]

The operator $\Hbb_0 + Q_N^{(0)}$ can be treated as an almost bosonic almost quadratic Hamiltonian; more precisely from \eqref{eq:functionaldef} and \eqref{eq:minenergy} we obtain a typical Bogoliubov-type energy
\begin{equation}\label{eq:minimum}
  \langle T \Omega, \left( \Hbb_0 + Q_N^{(0)} \right) T \Omega \rangle = \sum_{k \in \Zbbn} \frac{1}{2} \left( \sqrt{\alpha_k^2 - \beta_k^2} - \alpha_k \right) +  \varepsilon_1 + 2 \Re \varepsilon_2 \;.
\end{equation}
The errors $\varepsilon_1$ and $\varepsilon_2$ are controlled by Lemma \ref{lem:collectederrors}, based on estimating the number of excitations in the trial state $T\Omega$. Since $\alpha_k = \hbar \lvert k\rvert \mu + \beta_k + \mathcal{O}(N^{-1/2})$ and $\beta_k = \hbar \lvert k\rvert \mu^{-1} \hat{V}(k) + \mathcal{O}(N^{-1/2})$, we obtain the claimed formula for the energy.
\end{proof}

\section*{Acknowledgments}
It is a pleasure to thank Benjamin Schlein and Marcello Porta for discussions. The author was supported by the European Union through the ERC StG \textsc{FermiMath}, grant agreement nr.~101040991. Views and opinions expressed are however those of the author only and do not necessarily reflect those of the European Union or the European Research Council Executive Agency. Neither the European Union nor the granting authority can be held responsible for them. Moreover the author was partially supported by Gruppo Nazionale per la Fisica Matematica in Italy.

\newcommand{\etalchar}[1]{$^{#1}$}

\end{document}